\documentclass[10pt]{article}
\usepackage[utf8]{inputenc}
\usepackage[toc, page]{appendix}
\usepackage[font=small,labelfont=bf]{caption}
\usepackage{amsmath}
\usepackage{amssymb}
\usepackage{upgreek}
\usepackage{amsfonts} 
\usepackage{amsthm} 
\usepackage{color}
\usepackage{hyperref} 
\usepackage{setspace}
\usepackage{xcolor}
\usepackage{cancel}
\usepackage{enumerate}
\usepackage{Ari}
\usepackage{algorithm}          
\usepackage{algpseudocodex}     
\usepackage{nicematrix}

\usepackage{mdframed}

\hypersetup{colorlinks=true,linkcolor={red},citecolor={teal},urlcolor={teal}} 
\usepackage[total={6in,9in}]{geometry}
\usepackage{graphicx}
\graphicspath{{./images/}}

\newtheoremstyle{newstyle}      
{10pt} 
{10pt} 
{\itshape} 
{} 
{\bfseries} 
{.} 
{ } 
{} 

\theoremstyle{newstyle}
\newtheorem{theorem}{Theorem}[section]    
\newtheorem{definition}{Definition}[section]
\newtheorem{corollary}{Corollary}[section]

\newcommand{\secref}[1]{Section~\ref{#1}}
\newcommand{\defref}[1]{Definition~\ref{#1}}
\newcommand{\thmref}[1]{Theorem~\ref{#1}}

\newcommand{\corref}[1]{Corollary~\ref{#1}}

\newcommand{\algoref}[1]{Algorithm~\ref{#1}}

\renewcommand{\epsilon}{\varepsilon}

\let\oldnl\nl
\newcommand{\nonl}{\renewcommand{\nl}{\let\nl\oldnl}}

\title{Distributional PAC-Learning from Nisan's Natural Proofs\footnote{This is the full version; a conference version of this work is to appear at ITCS 2024.}}

\author{
Ari Karchmer\footnote{Boston University. Email: arika@bu.edu.}
}
\begin{document}

\maketitle

\begin{abstract}
    Do natural proofs imply efficient learning algorithms?
Carmosino et al. (2016) demonstrated that natural proofs of circuit lower bounds for $\Lambda$ imply efficient algorithms for learning $\Lambda$-circuits, but only over \textit{the uniform distribution}, with \textit{membership queries}, and provided $\AC^0[p] \subseteq \Lambda$. We consider whether this implication can be generalized to $\Lambda \not\supseteq \AC^0[p]$, and to learning algorithms which use only random examples and learn over arbitrary example distributions (Valiant's PAC-learning model). 

We first observe that, if, for any circuit class $\Lambda$, there is an implication from natural proofs for $\Lambda$ to PAC-learning for $\Lambda$, then standard assumptions from lattice-based cryptography do not hold. In particular, we observe that depth-2 majority circuits are a (conditional) counter example to the implication, since Nisan (1993) gave a natural proof, but Klivans and Sherstov (2009) showed hardness of PAC-learning under lattice-based assumptions. We thus ask:
what learning algorithms can we reasonably expect to follow from Nisan's natural proofs?

Our main result is that all natural proofs arising from a type of communication complexity argument, including Nisan's, imply PAC-learning algorithms in a new \textit{distributional} variant (i.e., an ``average-case'' relaxation) of Valiant's PAC model. Our distributional PAC model is stronger than the average-case prediction model of Blum et al. (1993) and the heuristic PAC model of Nanashima (2021), and has several important properties which make it of independent interest, such as being \textit{boosting-friendly}. The main applications of our result are new distributional PAC-learning algorithms for depth-2 majority circuits, polytopes and DNFs over natural target distributions, as well as the nonexistence of encoded-input weak PRFs that can be evaluated by depth-2 majority circuits.
\end{abstract}

\newpage

\pagenumbering{arabic}

\section{Introduction}

Razborov and Rudich \cite{razborov1997natural} introduced the concept of \textit{natural proofs} of circuit lower bounds.
Informally, a natural proof of a lower bound for a circuit class $\Lambda$ encodes an efficient algorithm that can be used to distinguish between the truth tables of \textit{simple} Boolean functions (those with ``small'' $\Lambda$-circuit complexity), and \textit{random} Boolean functions. Razborov and Rudich essentially showed that natural proofs for a circuit class $\Lambda$ rule out the existence of a cryptographic pseudorandom function (PRF) computable by $\Lambda$.

Carmosino et al. \cite{carmosino2016learning} strengthened the result of \cite{razborov1997natural} by demonstrating that, provided $\AC^0[p] \subseteq \Lambda$, natural proofs of circuit lower bounds for $\Lambda$-circuits of size up to $u(n)$ imply algorithms for \textit{learning} $\poly(n)$ size $\Lambda$-circuits with \textit{membership queries} over the \textit{uniform distribution}, in time exponential in $u^{-1}(\poly(n))$. 

As a corollary, \cite{carmosino2016learning} obtained a state-of-the-art \textit{quasipolynomial} time learning algorithm for $\AC^0[p]$-circuits (with membership queries, over the uniform distribution), using the natural proofs for subexponential size $\AC^0[p]$-circuits, for any prime $p$, of Razborov and Smolensky \cite{razborov1987lower, smolensky1987algebraic}. 

Since the result of \cite{carmosino2016learning}, whether or not there exists a fully general implication from natural proofs to learning algorithms in Valiant's \textit{original} PAC model \cite{valiant1984theory}, even for $\Lambda \not\supseteq \AC^0[p]$, has remained open (see e.g. \cite{goldberg2023improved}). In Valiant's original model, learning algorithms are forced to utilize \textit{random examples}, and learn over \textit{unknown example distributions}.

\begin{quote}
\textbf{Question 1.} Let $\Lambda$ be any circuit class. Do natural circuit lower bounds for size $u(n)$ $\Lambda$-circuits imply $\exp(u^{-1}(\poly(n)))$ time learning algorithms for $\poly(n)$ size $\Lambda$-circuits in Valiant's PAC model? 
\end{quote}

Aside theoretical interest in complexity and learning theory, Question 1 is motivated by the prospect of implicitly extending the nonexistence of PRFs in low circuit classes (derived from \cite{razborov1987lower,smolensky1987algebraic,razborov1997natural,carmosino2016learning}) to the nonexistence of \textit{weak} PRFs. A weak PRF is a PRF that is only required to be secure if the adversary can inspect uniformly random points, as opposed to to chosen points (see \secref{prelims:wprf} for a formal definition). Weak PRFs suffice for a variety of important cryptographic applications such as symmetric-key encryption (see e.g. \cite{boyle2021low} for more commentary). Therefore, understanding the minimum complexity needed to evaluate weak PRFs is of significant practical importance.

\subsection{Our Contributions}

We begin by observing that if the answer to Question 1 is essentially ``\textit{yes, for every }$\Lambda$,'' then this implies algorithmic breakthroughs for several important and well-studied computational problems. These breakthroughs include a classical polynomial time solution to the unique Shortest Vector Problem (uSVP), and quantum polynomial time algorithms for the Shortest Vector Problem (SVP) and Shortest Independent Vector Problem (SIVP) on lattices.\footnote{We will not try to discuss the huge literature on lattice problems (and lattice based cryptography). See \secref{prelims} for a short description of uSVP, SVP, and SIVP, and refer to \cite{regev2009lattices,regev2009complexity} for more information on complexity of lattice problems.} More specifically, we observe that majority-of-threshold circuits ($\majthr$) cannot realize the implication from natural proofs to polynomial time PAC-learning in Valiant's model, assuming polynomial time hardness of each of those problems.

\begin{theorem}\label{intro:natBarrier}
    Suppose that a natural proof against $\majthr$-circuits of size $u(n)$ implies that the class of $\majthr$-circuits of size $\poly(n)$ is PAC-learnable in Valiant's model, in time $\exp(u^{-1}(\poly(n)))$. Then, there is a polynomial time classical solution to $\tilde{O}(n^{1.5})$-$\mathrm{uSVP}$, and polynomial time quantum solutions to $\tilde{O}(n^{1.5})$-$\mathrm{SVP}$ and $\tilde{O}(n^{1.5})$-$\mathrm{SIVP}$.
\end{theorem}

To argue \thmref{intro:natBarrier}, we combine two observations. First, natural circuit lower bounds for $\exp(\Omega(n))$ size $\majthr$-circuits were proved by Nisan \cite{nisan1993communication}. Second, Klivans and Sherstov \cite{klivans2009cryptographic} showed hardness of polynomial time PAC-learning in Valiant's model for $\majthr$, assuming classical hardness of uSVP and quantum hardness SVP and SIVP. 
Taken together, we have both natural proofs against exponential-size $\majthr$, and hardness of Valiant's PAC-learning for each. Therefore, we have a circuit class that resists an implication between natural proofs and Valiant's PAC-learning. A formal argument is presented in \secref{section:NisanNatProp}; to the best of our knowledge, the natural property underlying Nisan's circuit lower bound for $\majthr$ has never been explicitly formalized (until \secref{section:NisanNatProp}), though it was acknowledged briefly by Raz \cite{raz2000bns} and considered implicitly by Viola \cite{viola2015communication}.

\thmref{intro:natBarrier} indicates a barrier to a general implication from a natural proof \textit{for any} $\Lambda$ to a PAC-learning algorithm for $\Lambda$ in Valiant's model. Essentially, the natural proof of \cite{nisan1993communication} for $\majthr$ confounds the hardness result of \cite{klivans2009cryptographic}.
In light of this, we shift our focus to the following more specific question.

\begin{quote}
    \textbf{Question 2.} What learning algorithms \textit{are} implied by Nisan's natural circuit lower bounds?
\end{quote}
Answering Question 2 is important if we want to gain understanding of a potential general implication between natural proofs for any class $\Lambda$, and some kind of learning algorithms for $\Lambda$. 

Towards an answer to Question 2, we will focus specifically on the possibility of learning algorithms that utilize only random examples, learn over unknown example distributions, and run in polynomial time.  
We briefly summarize our contributions towards an answer to Question 2, before digging into the specifics.

\begin{itemize}
    \item In \secref{introsec:distPAC}, we present a new learning model called \textit{distributional} PAC-learning, which relaxes Valiant's model, in order to try to sidestep \thmref{intro:natBarrier}.
    The new learning model is like Valiant's except it essentially removes the requirement of guarantees for the worst-case concept in the class. In \secref{section:distPAC_model}, we illustrate that distributional PAC-learning is independently motivated for both technically and practically oriented reasons.

    \item In \secref{section:NisanNatProp}, we prove \thmref{intro:natBarrier}. To do so, we give the first (to the best of our knowledge) explicit formalization of Nisan's natural property for exponential size $\majthr$-circuits.
    
    \item In \secref{section:mainTHM}, we prove our main theorem, which discovers a relationship between the computational complexity of distributional PAC-learning and the \textit{communication complexity} of a simple communication game that is associated with a given concept class. This theorem serves as a ``technical centerpiece'' for extracting distributional PAC-learning algorithms from Nisan's natural proofs for $\majthr$, and in fact an \textit{entire family} of natural proofs, of which Nisan's for $\majthr$ are a special case.
    
    \item In \secref{section:distPAC_algos} and \ref{section:NoWeakPRFs}, we use the main theorem and Nisan's natural proof for $\majthr$ to obtain new
    distributional PAC-learning algorithms for $\majthr$-circuits, polytopes, and DNFs. Additionally, we show how to derive attacks on any weak PRF evaluated by $\majthr$-circuits, even when the weak PRF is allowed an encoding procedure for the inputs.
\end{itemize}

We note that, because the result of \cite{carmosino2016learning} only applies to circuit classes that contain $\AC^0[p]$, prior to this work there were no known learning algorithms following directly from Nisan's natural proofs, in any nontrivial learning model.

\paragraph{Subsequent work.} In subsequent work, \cite{karchmer2023agnostic} demonstrated how to construct algorithms from Nisan's natural proofs in several other learning models. For example, \cite{karchmer2023agnostic} obtained ``nontrivial time'' agnostic membership query learning algorithms over the uniform distribution and ``nontrivial time'' distribution-independent membership query algorithms, for classes of sublinear size circuits made up of polynomial threshold function (PTF) gates and ${\sf SYM}^+$ gates.

\subsubsection{Distributional PAC-Learning}\label{introsec:distPAC}

With the goal of obtaining learning algorithms from Nisan's natural proofs in mind, this paper introduces the distributional PAC-learning model (distPAC-learning). 
However, the distPAC-learning model is also independently motivated as a relaxation of Valiant's PAC-learning, which we discuss after defining the model next.

The starting point for the distPAC-learning model is the heuristic PAC-learning (heurPAC-learning) model of Nanashima \cite{nanashima2021theory}. Following Nanashima, we define a Boolean concept class by a corresponding \textit{evaluation rule} $\eval = \{\eval_n: \{0,1\}^{*} \times \{0,1\}^n \rightarrow \{-1,1\}\}_{n \in\nat}$. The first input to the evaluation rule is a \textit{binary representation} $\pi_f$ of a concept $f$, and the second is an \textit{input} to the concept $x$. The evaluation rule is defined so that for every $n \in \nat$, $\eval_n(\pi_f, x) = f(x)$. An evaluation rule $\eval$ induces a Boolean concept class $\fC = \{\fC_n \}_{n\in\nat}$ defined by
\[
    \fC_n = \{f(x) := \eval_n(\pi_f,x) : \pi_f \in \{0,1\}^*\}
\]
We refer to $\fC$ as the $\eval$-induced concept class.

For a function $s: \nat \rightarrow \nat$, we say that the $\eval$-induced concept class $\fC$ is $s(n)$-\textit{represented} if, for every $n \in \nat$, under the evaluation rule $\phi_n$, every $f \in \fC_n$ has a binary representation of length at most $s(n)$. 
Considering evaluation rules helps for formalizing learning using a distribution over a concept class. We let $\mu$ denote a \textit{target distribution} over concepts $f \in \fC_n$, or, equivalently, over binary representations $\pi_f \in \{0,1\}^{s(n)}$.

What access to the concept algorithms in the distPAC-learning model are allowed? We continue following the heurPAC-learning model (Valiant's, too), and allow access to only random examples sampled from an \textit{unknown} example distribution $\rho$ over $\{0,1\}^n$. We denote by ${\ExO{f}{\rho}}$ the \textit{example oracle} that returns labelled examples $\langle x, f(x) \rangle$ for $x \sim \rho$.

A distPAC-learning algorithm takes three confidence parameters as
input. The accuracy parameter $\epsilon$, the failure parameter $\delta$, and 
the heuristic parameter $\eta$. Essentially, the distPAC-learning model requires that, for a fixed evaluation rule $\phi$, there exists some large probability mass of the $\eval$-induced concept class $\fC$, as determined by $\mu$ and $\eta$, that is learnable in Valiant's model.

\begin{definition}[Distributional PAC-learning]\label{def:distPAC}
    Let $\eval$ be an evaluation rule, and let the $\eval$-induced concept class $\fC$ be $s(n)$-represented. The pair $(\fC, \mu)$ is distributionally PAC-learnable if there exists an algorithm $A$ such that, for any $n \in \nat, \epsilon, \delta, \eta > 0$,
\begin{equation}\label{intro:avgPACguarantee}
    \Pr_{f \sim \mu}\left[\Pr_{A}\left[ \forall \rho : \Pr_{x \sim \rho}\left[h(x) \not= f(x) : h \leftarrow A^{\ExO{f}{\rho}}(n, \epsilon, \delta, \eta)\right] \le \epsilon \right] \ge 1-\delta \right] \ge 1-\eta
\end{equation}
When $A$ runs in time $\poly(n, s(n), \epsilon^{-1}, \delta^{-1}, \eta^{-1})$, we say that $(\fC, \mu)$ is efficiently distPAC-learnable. 
\end{definition} 

\noindent Distributional PAC-learning is a clear relaxation of Valiant's PAC-learning, since it no longer requires good learning guarantees for ``worst-case'' concepts.

\paragraph{DistPAC-learning vs. heurPAC-learning.} The essential difference between distPAC-learning and heurPAC-learning is the requirement that there exists a \textit{single} large probability mass of concepts that is learnable with respect to \textit{any} example distribution (see the location of the universal quantification over $\rho$ in \eqref{intro:avgPACguarantee}).
In heurPAC-learning, the order of quantifiers is different: it is only required that for each example distribution $\rho$, a large \textit{but possibly different} probability mass of the concept class is learnable. This independently motivates our model for technical and practical reasons:
for example, distPAC-learning allows ``black-box'' use of \textit{boosting algorithms}. In other words, the equivalence of weak\footnote{A weak learning algorithm is only required to output a hypothesis that has a predictive advantage only slightly better than a coin toss.} and strong learning is preserved in our model \cite{DBLP:journals/ml/Schapire90, domingo2000madaboost} (see \secref{section:distPAC_model} for a formal statement on this). This is not necessarily true in heurPAC-learning. 

In a nutshell, distPAC-learning is stronger than both heurPAC-learning and the Blum et al. average-case prediction model. Therefore distPAC-learning also inherits the well-founded motivation behind the theory of heuristic PAC-learning (see \cite{nanashima2021theory}). See \secref{comparison} for a continued discussion.
Towards further motivating the distPAC-learning model, in \secref{section:distPAC_model} we give formal statements on useful properties of the distPAC-learning model, including the equivalence between weak and strong learning, an equivalence between hardness of distribution-specific variant of distPAC-learning and the existence of one-way functions, and finally on using the classic technique of Kearns and Valiant \cite{kearns1994cryptographic} for proving hardness of distPAC-learning with respect to \textit{specific} target distributions. 

\subsubsection{DistPAC-Learning Algorithms from Nisan's Natural Proofs}

We design new learning algorithms in the distributional PAC model. These learning algorithms arise from natural proofs that follow a certain communication complexity argument due to Nisan \cite{nisan1993communication}.

\paragraph{Nisan's technique.} Nisan \cite{nisan1993communication} used the following communication complexity argument for proving circuit lower bounds against circuits with threshold functions as gates. First, identify a function $f: \{0,1\}^n \times \{0,1\}^n \rightarrow \{0,1\}$, which requires \textit{high} 2-party communication complexity in some model (e.g. randomized, determinstic, distributional, etc.). Then, identify a circuit class $\cC$ such that for every $g \in \cC$, $g$ is computable by a \textit{low-cost} 2-party communication protocol in that model. Finally, conclude that $f$ requires large $\cC$-circuits (see \secref{prelims} for essential definitions of 2-party communication complexity in various models, and \cite{kushilevitz1996communication} for further reference).
To provide an example, let $f(x,y) = {\sf IP2}(x,y) = \sum^n_{i=1}x_iy_i \mod 2$ be the inner product mod 2 function. 
It is known that ${\sf IP2}$ requires $\Omega(n)$ bits to be transmitted in any randomized communication complexity protocol; therefore, as shown by \cite{nisan1993communication}, since $\majthr$ circuits are computed by randomized communication complexity protocols with cost $O(\log n)$, ${\sf IP2}$ must require $\majthr$ circuits of exponential size. This lower bound remains one of the strongest known --- as of now it is still not ruled out that $\NEXP \subseteq \comp{\thr}{\thr}$. In fact, proving $\NEXP$ is not contained in $\comp{\thr}{\thr}$ is considered a ``major frontier'' in complexity theory \cite{chen2018toward}.

\paragraph{Main Theorem.}
We now introduce our main theorem, which is used as a ``technical centerpiece'' for obtaining distPAC-learning algorithms from natural circuit lower bounds proved with Nisan's technique. Concretely, we use it to obtain distPAC-learning algorithms for $\majthr$-circuits, which are presented after. 

Roughly speaking, the main theorem presents a relationship between the computational complexity of distPAC-learning, and the communication complexity of a simple \textit{communication game} associated with a given evaluation rule.
For any evaluation rule $\eval$ and $\eval$-induced $s(n)$-represented concept class $\fC$, we define the associated \textit{communication game} $\fG$ over the product distribution $(\mu, \rho)$, played as follows. A binary representation $\pi_f$ of a function $f \in \fC_n$ is sampled according to $\mu$, and an input $x$ is sampled from $\rho$. Player one is given the binary representation $\pi_f$ of $f$, and player two is given the input $x$. The two parties communicate until they are ready to output a value $b$, and win the game if $b = \eval_n(\pi_f,x) = f(x)$.
We say that $\eval$ is \textit{evaluated} by a 2-party distributional communication protocol with cost $c(n)$ and bias $\gamma(n)$ over $(\mu, \rho)$, if for every $n \in \nat$, the two parties can communicate at most $c(n)$ bits before winning $\fG$ with probability at least $1/2+\gamma(n)$ over the random inputs drawn from $(\mu, \rho)$. 

Now we are ready to state the main theorem. 

\begin{theorem}\label{intro:mainResult}
    Let $\eval$ be an evaluation rule. Suppose that, for any product distribution $(\mu, \rho)$, $\eval$ is evaluated by a 2-party distributional communication protocol with cost $c:= c(n)$ and bias $\gamma:= \gamma(n)$ over $(\mu, \rho)$. Then, for the $\eval$-induced $s(n)$-represented concept class $\fC$, and a time $t(n)$-samplable distribution $\mu$, the pair $(\fC, \mu)$ is distributionally PAC-learnable. The learning algorithm runs in time polynomial in $n, s(n), t(n), \epsilon^{-1}, \delta^{-1}, \eta^{-1}, \gamma^{-1}$ and $2^c$.
\end{theorem}

\noindent We give an overview of the proof of this theorem in \secref{proofIdeas}. We remark that slight improvement of the exponential dependency of $c$ requires a significant breakthrough in computational learning theory. For example, improving the efficiency to $O(2^{\sqrt{c}})$ gives a polynomial time distinguishing algorithm for the long-time weak PRF candidate of \cite{blum1993cryptographic}. We refer to \secref{sec:RequiresBreakthroughs} for details.

\paragraph{Interpretation of \thmref{intro:mainResult}.}
Nisan's lower bound technique clearly encodes a randomized communication complexity \textit{upper bound}, which, by an averaging argument, can be converted to a distributional protocol over any distribution (without increased cost or decreased bias).
Therefore, we interpret \thmref{intro:mainResult} as follows. 

Fix a function $f: \{0,1\}^n \times \{0,1\}^n \rightarrow \{0,1\}$, which requires $\Omega(n)$ bits in the 2-party randomized communication model (many such $f$ exist). Now, we interpret \thmref{intro:mainResult} as proof that a circuit lower bound for $f$ against $u(n)$ size $\Lambda$-circuits---\textit{proved by Nisan's method}---implies that the pair $(\fC, \mu)$ is distributionally PAC-learnable, in time $\exp(u^{-1}(\poly(n)))$, whenever $\eval \in \Lambda$, $\fC$ is a $\eval$-induced and $\poly(n)$-represented concept class, and as long as $\mu$ is polynomial time samplable. 

This is the valid interpretation because, since $f$ requires $\Omega(n)$ bits in the 2-party randomized communication model, any lower bound by Nisan's method against $u(n)$ size $\Lambda$-circuits \textit{requires} the existence of a 2-party randomized protocol to compute every $u(n)$ size $\Lambda$-circuit with $u^{-1}(\poly(n))$ bits. Hence, invoking \thmref{intro:mainResult} with these parameters plugged in, we get a distPAC-learner for $(\fC, \mu)$ that runs in time $\exp(u^{-1}(\poly(n)))$. Essentially, we have shown that Question 1 can be answered positively, if we insist on the natural proof being proved using Nisan's technique, and we weaken Valiant's PAC model to the distirbutional PAC model. We refer to \secref{addedRemarks} for a continued discussion.

\paragraph{DistPAC-learning from Nisan's natural proofs for $\majthr$.}
More concretely, we can use \thmref{intro:mainResult} to obtain \thmref{intro:distPAC_MajThr}, which directly follows from a combination of Nisan's lower bounds and \thmref{intro:mainResult}. This is because, as indicated in Nisan's lower bounds, every function in $\majthr$ has a randomized communication protocol of cost $O(\log n)$ and large bias.

\begin{theorem}\label{intro:distPAC_MajThr}
Let $\eval \in \majthr$ be any evaluation rule, and let $\mu$ be any polynomial time samplable target distribution. For the $\eval$-induced $s(n)$-represented concept class $\fC$, the pair $(\fC, \mu)$ is efficiently distPAC-learnable.
\end{theorem}

\thmref{intro:distPAC_MajThr} considers the concept class by the complexity of its evaluation rule. This is weaker than the learning-theoretic standard of considering the complexity of concepts directly. Any $s(n)$-represented concept class $\fC$ $\eval$-induced by the rule $\eval \in \cC$ must satisfy $\fC \subseteq \cC$ (assuming $\cC$-circuits can be $\poly(s(n))$-size). For $\cC$ containing a universal function (e.g. ${\sf P/poly}$, or $\NC^1$), $\cC = \fC$, but not necessarily for lower circuit classes such as $\majthr$.

By non-black-box inspection of the result of \cite{klivans2009cryptographic}, which we already mentioned proves hardness of learning $\majthr$ in Valiant's PAC model, we find that it actually provides a polynomial time samplable \textit{distribution} over $\majthr$-circuits that is hard to learn, even weakly (rather than just worst-case hardness). In other words, it shows a target distribution $\mu^*$ such that hardness of $\tilde{O}(n^{1.5})$-SVP and its variants implies $(\majthr, \mu^*)$ is not efficiently distPAC-learnable! In \secref{sec:hard_distPAC_crypto}, we actually show that the \textit{entire family} of proof techniques for showing hardness of Valiant's PAC-learning (used here by \cite{klivans2009cryptographic}, and due originally to \cite{kearns1994cryptographic}), can be used to prove hardness of distPAC-learning with respect to \textit{specific} polynomial time samplable target distributions.

Therefore, in \thmref{intro:distPAC_MajThr}, considering the complexity of the evaluation rule $\eval \in \majthr$ and distPAC-learning of the $\eval$-induced concept class $\fC$ is likely a needed relaxation, since the target distribution in the theorem can be \textit{any} polynomial time samplable distribution. Hence, in order to get a distPAC-learning algorithm for $\majthr$-circuits, we need to restrict the target distribution somehow.

\paragraph{Natural Target Distributions.} In light of this, we show that $\majthr$-circuits are distPAC-learnable, with respect to the following more natural families of target distributions. By more natural, we mean that the target distribution is not designed by a cryptographer (as opposed to $\mu^*$). Indeed, this highlights another feature of the distPAC-learning model: efficient learnability of $(\fC, \mu)$ for ``organic'' target distributions $\mu$ can coexist with hardness for ``inorganic'' target distributions like $\mu^*$.

Let $L = (T_1, \cdots T_{m})$ be a list of $m := \poly(n)$ linear threshold functions, and let $\mu$ be a $\poly(n)$ time samplable distribution over $\{0,1\}^{m}$. We define the distribution $\mu_L$ over $\majthr$-circuits is sampled as follows. First, sample $\theta \sim \mu$. Then, output the $\majthr$-circuit that is the majority vote over each $T_i \in L$ such that $\theta_i = 1$.

\begin{theorem}\label{intro:distPAC_MajThr2}
    Let $L = (T_1, \cdots T_{m})$ be any list of $m := \poly(n)$ linear threshold functions, and let $\mu$ be any $\poly(n)$ time samplable distribution over $\{0,1\}^{m}$. The pair $(\majthr, \mu_L)$ is efficiently distPAC-learnable.
\end{theorem}

\noindent Previously, no polynomial time distPAC-learning algorithms were known, for any reasonable type of target distributions over $\majthr$-circuits. An interesting feature of our distPAC-learning algorithm is that it does not need to know $\mu$ or $L$ to work (see end of \secref{subsubsec:nat-dist-majthr} for details).

We additionally consider slightly modified target distributions, that correspond to interesting and natural distributions over subclasses of $\majthr$-circuits: polytopes (that is, and-of-thresholds circuits ($\comp{\AND}{\thr}$)) and DNFs. Note that, in distributional PAC-learning, subclasses are not necessarily distPAC-learnable if their superclass is, since it is possible that the subclass is hard-core, and consisting of functions that are hard for the superclass distPAC-learning algorithm.

For the polytope distribution, let $L = (T_1, \cdots T_{m})$ be a list of $m := \poly(n)$ linear threshold functions, and let $\mu$ be a $\poly(n)$ time samplable distribution over $\{0,1\}^{m}$. The distribution $\mu^\land_L$ over polytopes is sampled as follows. First, sample $\theta \sim \mu$. Then, output the polytope that is the conjunction of all $T_i \in L$ such that $\theta_i = 1$.

\begin{theorem}\label{intro:polytope}
    Let $L = (T_1, \cdots T_{m})$ be any list of $m := \poly(n)$ linear threshold functions, and let $\mu$ be any $\poly(n)$ time samplable distribution over $\{0,1\}^{m}$. The pair $(\comp{\AND}{\thr}, \mu^\land_L)$ is efficiently distPAC-learnable.
\end{theorem}

For the DNF distribution, let $L = (T_1, \cdots T_{m})$ be a list of $m := \poly(n)$ disjunctions on $n$-bit inputs, and let $\mu$ be a $\poly(n)$ time samplable distribution over $\{0,1\}^{m}$. The distribution $\mu^{\land\lor}_L$ over DNFs is sampled as follows. First, sample $\theta \sim \mu$. Then, output the DNF that is a conjunction of all disjunctions $T_i \in L$ such that $\theta_i = 1$.

\begin{theorem}\label{intro:DNF}
    Let $L = (T_1, \cdots T_{m})$ be any list of $m := \poly(n)$ disjunctions on $n$-bit inputs, and let $\mu$ be any $\poly(n)$ time samplable distribution over $\{0,1\}^{m}$. The pair $(\mathrm{DNF}, \mu^{\land\lor}_L)$ is efficiently distPAC-learnable.
\end{theorem}

Even though distPAC-learning is stronger that heurPAC-learning, \thmref{intro:distPAC_MajThr2}, \ref{intro:polytope} and \ref{intro:DNF} are formally incomparable to the heurPAC-learning algorithm for $O(\log n)$-juntas due to \cite{nanashima2021theory}. This is for the following reasons. On one hand, \thmref{intro:distPAC_MajThr2}, \ref{intro:polytope} and \ref{intro:DNF} are stronger because $\majthr$-circuits, polytopes, and DNFs are strictly more powerful than $O(\log n)$-juntas, and we handle learning over arbitrary example distributions, while \cite{nanashima2021theory} only handles the uniform example distribution. However, the confounding variable is that the heurPAC algorithm works with respect to the uniform distribution over $O(\log n)$-juntas, while for any $L, \mu$, the target distributions that we learn are not uniform over their support. We thus cannot show that our algorithm is stronger than Nanashima's in a formal sense.
An interesting direction for future work is to obtain distPAC-learning algorithms for other natural distributions over $\majthr$-circuits, polytopes, and DNFs.

\subsubsection{Impossibility of Encoded-Input Weak PRFs}

Although we do not obtain any PAC-learning algorithm in Valiant's model from Nisan's natural proofs, we show that distributional PAC-learning is still enough to rule out weak PRFs (which was one of the initial motivations of studying Question 1). In fact, we show that this is true even when the weak PRF is allowed an arbitrary input encoding.

\begin{theorem}\label{intro:NoEIWPRF}
    There exists no encoded-input weak PRF that is evaluated by a $\majthr$-circuit.
\end{theorem}

Our notion of encoded-input weak PRF is the natural weak analogue of the encoded-input PRF introduced by \cite{boneh2018exploring}. Loosely speaking, an encoded-input weak PRF is a PRF that is only required to be secure when the adversary sees random points, where the inputs are taken uniformly at random from a predefined multi-subset of the input space. We refer to \secref{section:NoWeakPRFs} for details.

\subsection{Proof Overview of \thmref{intro:mainResult}}\label{proofIdeas}

We now overview the ideas behind the proof of \thmref{intro:mainResult}. 
The most important tool we use is the $2$-party norm  of a function, $R_2(f)$, which is defined to be the expected product of a function computed on a list of correlated inputs. 

\begin{definition}[$2$-party norm]
For $f: (\{0,1\}^{n})^2 \rightarrow \{-1,1\}$, the $2$-party norm of $f$ is defined as
\begin{equation}\label{intro:2partynorm}
    R_2(f) := \Ex{x^0_1, x^0_2, x^1_1, x^1_2 \sim U_n}{\prod_{\epsilon_1,\epsilon_2 \in \{0,1\}} f(x^{\epsilon_1}_1, x^{\epsilon_2}_2)}
\end{equation}
\end{definition}

Throughout the paper, we use $U_n$ to denote the uniform distribution over $\{0,1\}^n$.
The 2-party norm is a special case of the $k$-party norm (sometimes called the cube-measure or box-norm), which was introduced by \cite{babai1992multiparty} for obtaining lower bounds in $k$-party Number-on-Forehead communication complexity.

The crucial property about $R_2(f)$ is that, up to parameters, it upper bounds the correlation of $f$ with functions computable by deterministic $2$-party communication protocols. We denote by $\dCC{2}{c}$ the set of all $f: (\{0,1\}^{n})^2 \rightarrow \{-1,1\}$ that have deterministic 2-party communication protocols with cost at most $c$. For a definition of the deterministic $2$-party communication model, see \secref{prelims:kCC}.

Implicit in all three of \cite{chung1993communication, raz2000bns, viola2007norms} (who showed a related theorem in the more general $k$-party case), is the following bound: 
\begin{theorem}[\textit{The} correlation bound --- \cite{chung1993communication,raz2000bns,viola2007norms}]
\label{intro:thm:corrBound}
For every function $f : (\{0,1\}^{n})^2 \rightarrow \{-1,1\}$, 
\begin{equation}\label{intro:corrBound}
    \corr{\dCC{2}{c}} = \max\limits_{\pi \in \dCC{2}{c}}\left|\Ex{x}{f(x) \cdot \pi(x)}\right| \le 2^c \cdot R_2(f)^{1/4}
\end{equation}
for $x$ uniformly distributed over $(\{0,1\}^{n})^2$.
\end{theorem}

\noindent Equation \eqref{intro:corrBound} implies that $(2^{-c} \cdot \corr{\dCC{2}{c}})^4 \le R_2(f)$.

The construction of the learning algorithm of \thmref{intro:mainResult} uses the lower bound on $R_2(f)$ to distinguish structure from randomness. In other words, hypothetically consider functions $f: (\{0,1\}^{n})^2 \rightarrow \{-1,1\}$ such that the quantity $(2^{-c} \cdot \corr{\dCC{2}{c}})^4$ is relatively \textit{large} (greater than $1/\poly(n)$, say). Such functions can be distinguished from uniformly random functions, by taking a random sample from the distribution over the value inside the expectation in \eqref{intro:2partynorm}. This follows from the fact that $R_2(\psi)$ for a uniformly random function $\psi : (\{0,1\}^{n})^2 \rightarrow \{-1,1\}$ is bounded from above by a negligible function of $n$.

Using this idea, we have the following proof outline. First, we can try to prove a ``distinguisher-to-predictor'' lemma, in the style of \cite{yao1982theory}, in order to obtain a weak randomized predictor for $f$ (a weak predictor requires accuracy of a prediction for an unseen example to be only slightly more accurate than a coin toss). Second, we could apply standard averaging arguments to construct a weak PAC-learning algorithm. Finally, we could apply celebrated boosting results from learning theory \cite{DBLP:journals/ml/Schapire90, domingo2000madaboost} to produce a full-blown PAC-learning algorithm.

However, this proof outline remains incomplete. First, the 2-party norm of the function is the expectation of a product of \textit{correlated} inputs, so we have not given any way of using independent random examples. Second, we have said nothing of how to handle arbitrary example distributions (the inputs to $f$ on the right hand side of \eqref{intro:corrBound} should be uniformly random). 
We handle both of these problems simultaneously, roughly by thinking of $f$ as the \textit{evaluation rule}, and not the concept itself. 

First, let us describe how $f$ should be viewed in more detail. There are two inputs to $f$, $x_1$ and $x_2$. Without loss of generality, identify $x_1$ as a random string for sampling the target distribution $\mu$, and identify $x_2$ as a random string for sampling the example distribution $\rho$, with $|x_1| = |x_2| = m$. We abuse the notation and let $z =\rho(x_2)$ to denote a point $z$ sampled according to $\rho$ with the random bits $x_2$. Similarly, we let $g$ be the function represented by $\pi_g = \mu(x_1)$. Next, fix the evaluation rule $\eval$, which is the map that takes as input the concept representation $\pi_g$, plus the input $z$, and outputs $\eval(\pi_g,z) = g(z) = y$. As a function of $x_1,x_2$, we can thus write the process of generating a labelled example as $\langle \rho(x_2), \eval(\mu(x_1), \rho(x_2)) \rangle = \langle z,g(z)\rangle = \langle z,y\rangle$. We let $f(x_1,x_2) =\eval(\mu(x_1), \rho(x_2))$. This allows us to write:

\begin{equation*}
    R_2(f) = \Ex{x^0_1, x^0_2, x^1_1, x^1_2}{v(x^0_1, x^0_2, x^1_1, x^1_2)} \text{\ for \ } v(x^0_1, x^0_2, x^1_1, x^1_2) := \prod_{\epsilon_1,\epsilon_2 \in \{0,1\}} \eval(\mu(x^{\epsilon_1}_1), \rho(x^{\epsilon_2}_2))
\end{equation*}

Now, we describe how we construct a weak randomized predictor which only uses random examples from an arbitrary $\rho$. At the core, we will use the example oracle to sample a single instance of $v(x^0_1, x^0_2, x^1_1, x^1_2)$, over uniformly random $x^0_1, x^0_2, x^1_1, x^1_2 \in \{0,1\}^m$. To see the significance of this, observe that by definition $v(x^0_1, x^0_2, x^1_1, x^1_2)$ has expected value $R_2(f)$. Hence, the process of sampling this value distinguishes examples labelled by uniformly random functions from examples labelled by concepts sampled according to $\mu$ --- as long as $\mu$ samples representations of concept that are evaluated by $\eval$. This claim is justified because whenever it is possible to win the communication game $\fG$ associated with $\eval$ with high bias and low communication, \thmref{intro:thm:corrBound} implies that $R_2(f)$ is large. In other words, $R_2(f)$ is guaranteed to be large whenever it is possible to efficiently (probabilistically) communicate the evaluation rule $\eval$ (because this implies winning $\fG$ with good bias). At this point, we use a simple hybrid argument to construct a randomized prediction algorithm for examples sampled according to $\rho$.

It remains to verify that the randomized prediction algorithm can actually sample $v(x^0_1, x^0_2, x^1_1, x^1_2)$, using only access to $\ExO{g}{\rho}$, where $g$ is the concept sampled according to the target distribution $\mu$. To see this, observe that the distribution over $v(x^0_1, x^0_2, x^1_1, x^1_2)$ is identical to the distribution over $g(z_1)g(z_2)h(z_1)h(z_2)$, for $\langle z_1, g(z_1) \rangle, \langle z_2, g(z_2)\rangle \sim \ExO{g}{\rho}$, and $h \sim \mu$. The value $h(z_1)h(z_2)$ can be computed because $h(z_1)$ and $h(z_2)$ can be \textit{queried}, since $h$ is sampled \textit{locally} by the algorithm. 
Therefore, we only need $\ExO{g}{\rho}$.

We also need to verify that $\rho$ need not be efficiently samplable. To argue this this, we observe that communicating parties participating in $\fG$ have unbounded computational power. This means that, even if $\rho$ is an arbitrary distribution, there is no effect on the communication cost of $\fG$. Indeed, the process of sampling $\rho$ can be viewed as a local pre-processing step in the protocol for party two. Therefore, $R_2(f)$ does not decrease when $\rho$ is arbitrary.

\subsection{Discussion}

\subsubsection{Distributional PAC-Learning vs. Related Models}\label{comparison}

As mentioned previously, distPAC-learning strengthens heurPAC-learning. This is due to the stronger quantification over example distributions. The main benefit of this is that it facilitates boosting of weak learning algorithms, which needs worst-case guarantees over the example distribution (see \secref{weak_strong_equiv} for a formal statement). We encourage the reader to visit Section 1.2 of \cite{nanashima2021theory}, as their points regarding the motivations of heurPAC-learning as a relaxation of Valiant's PAC model, largely apply to distPAC-learning as well. Additionally, see Section 1.2 of \cite{nanashima2021theory} for commentary of the differences with previous ``implicit'' definitions of average-case learning, such as in \cite{jackson2005learning, jackson2011learning, sellie2009exact} also apply to distPAC-learning. 

In comparison to the seminal work of \cite{blum1993cryptographic}, distPAC-learning also differs on the order of quantifiers. In the definition of the average-case prediction considered by \cite{blum1993cryptographic}, both the target distribution $\mu$ and the example distribution $\rho$ are fixed. This means that there can be a different prediction algorithm, for each pair of $\mu$ and $\rho$. This model is weaker than both the heurPAC-learning and distPAC-learning models.

\subsubsection{Other Related Work}

Many other relationships between learning theory and communication complexity have been studied. Some notable examples include \cite{kremer1999randomized, linial2009learning, feldman2014sample, kane2019communication} (also see the references therein). All of these works study relationships between communication complexity and notions of learning complexity, such as sample complexity \cite{kremer1999randomized, kane2019communication}, differentially private sample complexity \cite{feldman2014sample}, margin complexity \cite{linial2009learning}, VC dimension \cite{kremer1999randomized, feldman2014sample} and Littlestone dimension \cite{feldman2014sample}. These works are all incomparable to ours, as they do not directly study relationships between communication complexity and the \textit{computational} complexity of learning.

Learning intersections of halfspaces (i.e., ands of linear threshold functions) was considered by \cite{klivans2004learning}. Using Fourier-analytic techniques, \cite{klivans2004learning} showed a polynomial time learning algorithm for any function of a constant number of halfspaces with respect to the uniform distribution over examples. Additionally, \cite{klivans2004learning} gave a quasi-polynomial time algorithm for learning any Boolean function of a \textit{polylogarithmic} number of bounded-weights linear threshold functions, under \textit{any} distribution over examples. Our learning results (\thmref{intro:distPAC_MajThr}, and \thmref{intro:distPAC_MajThr2}) are at the moment similar but incomparable; we get polynomial time \textit{distributional} PAC-learning of concepts \textit{evaluated} by majorities of linear threshold functions over \textit{any example distribution}.

\subsubsection{Additional Remarks and Future Work}\label{addedRemarks}


In this work, we began by observing that if Question 1 resolves to ``yes, for every $\Lambda$,'' then cryptographic assumptions such as quantum polynomial time hardness of $\tilde{O}(n^{1.5})$-SVP (and its variants) do not hold. 
To continue our study, we shifted the focus to understanding what kind of learning algorithms \textit{are} implied, specifically by Nisan's communication complexity based natural proof technique (Question 2). 

To this end, we introduced the distPAC-learning model as a relaxation of Valiant's PAC model, and then proved \thmref{intro:mainResult}. We interpret \thmref{intro:mainResult} as proof that Question 1 can be answered positively, if we insist that the natural circuit lower bound is proved using Nisan's technique specifically, and the learning model is the distributional PAC model rather than Valiant's worst-case PAC model.



Towards this result, we exploited the specific aspects of Nisan's lower bound method. In other words, we did not simply use the fact that they were natural, but specifically \textit{how} they are natural. 
Therefore, it remains open whether or not \textit{other} natural proofs imply efficient distPAC-learning algorithms for \textit{other} concept classes, such as $\AC^0[p]$. 
In fact, the natural proofs for $\AC^0[p]$ of \cite{razborov1987lower, smolensky1987algebraic} are not affected by \thmref{intro:natBarrier}, so it is possible they could even imply algorithms in Valiant's model. At present, the difficulty in proving a similar barrier to \thmref{intro:natBarrier} for $\AC^0[p]$ is that, while we are able to prove hardness of PAC-learning in Valiant's model, we do not have any exponentially strong natural proofs.

We believe that the primary direction for future research is use the Razborov-Smolensky lower bounds in a non-black box way, in order to obtain distPAC-learning or even PAC-learning for $\AC^0[p]$. We note that works such as \cite{boyle2021low} have introduced conjectured weak PRF candidates that can be evaluated by $\AC^0[2]$, with considerable evidence to support \textit{subexponential} security of the candidates. A weak PRF evaluated by $\AC^0[2]$ with subexponential security would preclude any (even quasipolynomial time) distPAC-learning algorithm for concept classes induced by a $\AC^0[2]$ evaluation rule. That being said, we hope that our results shed light on what aspects of natural proofs are useful for learning algorithms that cannot query the concept.



\section{Preliminaries and Definitions}\label{prelims}

\subsection{2-Party Communication Complexity and Norms}\label{prelims:kCC}


In the following, we discuss Boolean functions that output -1 or 1. 


The $2$-party communication model is the following. There are $2$ parties, each having unbounded computational power, who try to collectively compute a function. The input to the function is separated into $2$ segments, and the $i^{th}$ party sees the $i^{th}$ segment. The parties can send each other direct messages.

Each party may transmit messages according to a fixed protocol. The protocol determines, for every sequence of bits transmitted up to that point (the transcript), whether the protocol is finished (as a function of the transcript), or if, and which, party writes next (as a function of the transcript) and what that party transmits
(as a function of the transcript and the input of that party). Finally, the last bit transmitted is the output of the protocol, which is a value in $\{-1, 1\}$. The complexity measure of the protocol is the total number of bits transmitted by the parties.

\begin{definition}[$\dCC{2}{c}$ class]
$\dCC{2}{c}$ is defined to be the class of functions $f: (\{0,1\}^n)^2 \rightarrow \{-1,1\}$ that can be computed by a $2$-party deterministic communication protocol with complexity $c$.
\end{definition}

Another communication model is \textit{randomized} communications.

\begin{definition}[Randomized $\dCC{2}{c}$]
The randomized $2$-party communication model allows the protocol to depend on random bits. Therefore, we allow the protocol to err in its output. The probability of error of a randomized protocol is $\epsilon$ if for every input to the function $f$, the protocol errs in outputs with probability at most $\epsilon$. We denote by $\randCC{2}$ the class of $2$-party randomized protocols that transmit at most $c$ bits and err with probability at most $1/2- \gamma$.

For the sake of simplicity, this paper uses only the public coin version of randomized communication complexity. Namely the parties all share a string of random bits.
\end{definition}

A model more relaxed than randomized communication is \textit{distributional} communication.

\begin{definition}[Distributional $\dCC{2}{c}$]
The distributional $2$-party communication model allows the protocol to err on certain inputs. Fix a distribution $\rho$ over $(\{0,1\}^n)^2$. A function $f: (\{0,1\}^n)^2 \rightarrow \{-1,1\}$ is in $\distCC{2}$ if there exists a communication protocol $p \in \dCC{2}{c}$ such that
\[
\Ex{(x_1, x_2) \sim \rho}{p(x_1, x_2) \cdot f(x_1, x_2)} \ge 2\gamma
\]
\end{definition}

Distributional communication complexity can be thought of as correlation.

\begin{definition}[Boolean function correlation]
    Define $\corr{\Lambda} := \max_{h \in \Lambda}|\ex{f(x) \cdot h(x)}|$, where $x$ is sampled uniformly at random from the domain.
\end{definition}

When we want to measure correlation between two function classes, we have it defined as follows:

\begin{definition}[Boolean function correlation]
    Define $\corrClass{\fC}{\Lambda} := \min_{f \in \fC}\max_{h \in \Lambda}|\ex{f(x) \cdot h(x)}|$, where $x$ is sampled uniformly at random from the domain.
\end{definition}

When $\rho$ is the uniform distribution, $f \in \distCC{2}$ is equivalent to $\corr{\dCC{2}{c}} \ge 2\gamma$.

A simple fact is that for any distribution $\rho$, $f \in \randCC{2}$ implies that $f \in \distCC{2}$. Therefore, $f \in \randCC{2}$ implies that $\corr{\dCC{2}{c}} \ge 2\gamma$.

\begin{definition}[$2$-party norm]
For $f: (\{0,1\}^n)^2 \rightarrow \{-1,1\}$, the $2$-party norm of $f$ is defined as
\begin{equation}\label{2partynorm}
    R_2(f) := \Ex{x^0_1, x^0_2, x^1_1, x^1_2 \sim \{0,1\}^n}{\prod_{\epsilon_1,\epsilon_2 \in \{0,1\}} f(x^{\epsilon_1}_1, x^{\epsilon_2}_2)}
\end{equation}
\end{definition}

The 2-party norm is a special case of the $k$-party norm (sometimes called the cube-measure), which was introduced by \cite{babai1992multiparty} for obtaining lower bounds in $k$-party Number-on-Forehead communication complexity.

The crucial property about $R_2(f)$ is that, up to parameters, it upper bounds the correlation of $f$ with functions computable by $2$-party communication protocols. Implicit in all three of \cite{chung1993communication, raz2000bns, viola2007norms} (who showed a related theorem in the more general $k$-party case), is the following bound: 
\begin{theorem}[\textit{The} correlation bound --- \cite{chung1993communication,raz2000bns,viola2007norms}]
\label{thm:corrBound}
For every function $f : (\{0,1\}^n)^2 \rightarrow \{-1,1\}$, 
\begin{equation}\label{corrBound}
    \corr{\dCC{2}{c}} \le 2^c \cdot R_2(f)^{1/4}
\end{equation}
\end{theorem}

An immediate corollary of this bound is: 
\begin{theorem}
\label{thm:randBound}
For every function $f : (\{0,1\}^n)^2 \rightarrow \{-1,1\}$, such that $f \in \RANDCC{2}{c}{\gamma}$,
\begin{equation}
    \gamma \le 2^c \cdot R_2(f)^{1/4}
\end{equation}
\end{theorem}

\subsection{(Weak) Pseudorandom Functions}\label{prelims:wprf}

For clarity, we define weak and strong pseudorandom functions. See \secref{section:NoWeakPRFs} for the definition of encoded-input weak PRFs.

\begin{definition}[Weak and Strong PRFs]
    Let $\lambda$ be a security parameter, and $n=n(\lambda),\kappa = \kappa(\lambda)$ for polynomially bounded functions $n, \kappa$. Consider a pair of algorithms ${\sf f}: \{0,1\}^{\kappa} \times \{0,1\}^n \rightarrow \{0,1\}^1, {\sf gen}: \{1\}^\lambda \rightarrow \{0,1\}^{\kappa}$. 
    \begin{itemize}
        \item ${\sf gen}$ is a polynomial time sampling algorithm that given input parameter $\lambda$ in unary and access to random coins $z \in \{0,1\}^{\poly(\lambda)}$ outputs a key $k \in \{0,1\}^{\kappa}$.
        \item ${\sf f}$ is a polynomial time algorithm that given a key $k$ and the input $x \in \{0,1\}^n$, outputs a value ${\sf f}(k,x) = v \in \{0,1\}^1$.
    \end{itemize} 
    For $t=t(\lambda), \epsilon=\epsilon(\lambda)$, we say that $({\sf f}, {\sf gen})$ is a $(t,\epsilon)$-weak PRF if, for every size $t$ oracle circuit $C$,
    \[
    \left|\Pr_{k \sim {\sf gen}(1^\lambda)}\left[C^{\ExO{{\sf f}(k, \cdot)}{U_n}}=1\right]- \Pr_{r}\left[C^{\ExO{r}{U_n}}=1\right]\right| \le \epsilon(\lambda)
    \]
    where $r: \{0,1\}^n \rightarrow \{0,1\}$ is a uniformly random function.  

    Additionally, we say that $({\sf f}, {\sf gen})$ is a $(t,\epsilon)$-PRF if, for every size $t$ oracle circuit $C$,
    \[
    \left|\Pr_{k \sim {\sf gen}(1^\lambda)}\left[C^{{\sf f}(k, \cdot)}=1\right]- \Pr_{r}\left[C^{r}=1\right]\right| \le \epsilon(\lambda)
    \]
    where $r: \{0,1\}^n \rightarrow \{0,1\}$ is a uniformly random function.
\end{definition}

When $t, \epsilon = 2^{\log ^c(n)}$ for some constant $c > 1$, we say that $({\sf f}, {\sf gen})$ has \textit{quasipolynomial} security.
When $t, \epsilon = 2^{\lambda^\delta}$ for some constant $\delta \in (0,1)$, we say that $({\sf f}, {\sf gen})$ has \textit{subexponential} security.

\subsection{Circuit Classes and Other Computational Classes}\label{prelims:circuit-classes}

We will consider various circuit classes with different bases (all being defined previously in the literature). $\AC^0$ is the class of constant depth, polynomial size, unbounded fan-in ${\sf AND/OR/NOT}$ circuits. $\AC^0[p]$ is the class of constant depth, polynomial size, unbounded fan-in ${\sf AND/OR/NOT/MOD}p$ circuits, where $p \in \nat$ is a prime number. $\TC^0$ is the class of constant-depth, polynomial size, unbounded fan-in circuits of $\thr$ gates, where a $\thr$ gate is a linear threshold function $t(x_1, \cdots x_m) := \sum^m_{i=1} w_ix_i \ge^? \theta$, which outputs 1 if and only if the sum of the inputs weighted by real coefficients $w_1, \cdots w_m$ exceeds a threshold $\theta$. When the weights are fixed to be $1$ and $\theta = m/2$, we call it a $\maj$ gate. An ${\sf XOR}$ gate takes the sum modulo 2 of its inputs.

Many circuit classes considered are of the form $\comp{\cC}{\cC'}$ for circuit classes $\cC, \cC'$. The composed class $\comp{\cC}{\cC'}$ denotes the class of $\cC$-circuits with inputs as the outputs of functions from $\cC'$. For example, the class $\comp{\thr}{\thr}$ (a.k.a. depth-2 $\TC^0$). Alternatively, $\comp{\maj}{\thr}$ is the class of circuits consisting of a $\maj$ gate composed with a bottom layer of $\thr$ gates.

\subsection{Lattice Problems}


An $n$-dimensional lattice $L$ is the set of all integer linear combinations of a given real basis $v_1 \cdots v_n \in \real$. Mathematically, 
\[
    L = \left\{\sum\limits_{i \in [n]} c_iv_i : c_1, \cdots c_n \in \ints\right\}
\]
In the unique shortest vector problem, $f(n)$-uSVP, one must find the shortest non-zero vector in the lattice (if there are many, then any is valid), with the promise that the shortest one is shorter, by at least an $f(n)$ factor, than any other non-parallel vector in the lattice. In the Shortest Vector Problem (SVP), $f(n)$-SVP, one must approximate the length of the shortest non-zero vector in the lattice within a factor of $f(n)$. Lastly, in the Shortest Independent Vector Problem ($f(n)$-SIVP), the object is to find a collection of $n$ linearly independent lattice vectors of length at most $f(n) \cdot a$, where $a$ is the minimum length of the longest vector present in some subset of $n$ vectors in the lattice that is linearly independent.

In general, the hardness of uSVP, SVP, and SIVP increases as the parameter $f(n) \le \poly(n)$ decreases. This work uses, as a black box, hardness results for PAC-learning $\majthr$-circuits due to Klivans and Sherstov \cite{klivans2004learning}. Klivans and Sherstov in turn rely in a black box way on public key cryptosystems of Regev \cite{regev2004new} to prove their hardness results. The cryptosystems of Regev rely on polynomial time hardness of uSVP, and quantum polynomial time hardness of SVP and SIVP.

\section{The Distributional PAC-Learning Model}\label{section:distPAC_model}

In this section, we give some useful properties of distPAC-learning, which motivate the definition independently. First, we cover the equivalance between weak and strong versions of distPAC-learning. Then, we show that distPAC-learning implies heurPAC-learning. Finally, we consider the realationship of distPAC-learning with cryptography, by demonstrating that hardness of distPAC-learning for polynomial size circuits in a distribution-specific setting is equivalent to the existence of (infinitely-often) one-way functions. Finally, we show that the famous technique of \cite{kearns1994cryptographic} for proving hardness of Valiant's PAC-learning using public key encryption schemes generates ``hard target distributions'' for distPAC-learning.

\subsection{Equivalence Between Weak and Strong DistPAC-Learning}\label{weak_strong_equiv}

Celebrated results of computational learning theory, indicate that efficient weak and strong PAC-learning are equivalent \cite{DBLP:journals/ml/Schapire90} (in the ``filtering'' setting, this is shown by e.g. \cite{domingo2000madaboost}). 

Recall that $\fC = \{\fC_n\}_{n \in \nat}$ is a Boolean concept class that is $s(n)$-represented and induced by the evaluation rule $\eval = \{\eval_n\}_{n \in \nat}$. We define a weak version of distPAC-learning. The only difference with \defref{def:distPAC} is that the accuracy of the hypothesis only needs to achieve error below $1/2 - \epsilon$ with high probability.

\begin{definition}[Weak distPAC-learning]\label{def:weakDistPAC}
    Let $\eval$ be an evaluation rule, and let the $\eval$-induced concept class $\fC$ be $s(n)$-represented. The pair $(\fC, \mu)$ is weakly distributionally PAC-learnable if there exists an algorithm $A$ such that, for any $n \in \nat, \delta, \eta > 0$,
\begin{equation}
    \Pr_{f \sim \mu}\left[\Pr_{A}\left[ \forall \rho : \Pr_{x \sim \rho}\left[h(x) \not= f(x) : h \leftarrow A^{\ExO{f}{\rho}}(n, \delta, \eta)\right] \le 1/2-1/\poly(n) \right] \ge 1-\delta \right] \ge 1-\eta
\end{equation}
When $A$ runs in time $\poly(n, s(n), \delta^{-1}, \eta^{-1})$, we say that $(\fC, \mu)$ is efficiently distPAC-learnable. 
\end{definition}


We now demonstrate that the boosting results from Valiant's PAC model carry over to the distPAC-learning model.

\begin{theorem}[distPAC Boosting]\label{distPAC_weak_strong_equiv}
Let $\eval$ be an evaluation rule, and let $\fC$ be the $\eval$-induced $s(n)$-represented concept class. For any target distribution $\mu$,
$(\fC, \mu)$ is efficiently weakly distributionally PAC-learnable if and only if $(\fC, \mu)$ is efficiently distributionally PAC-learnable.
\end{theorem}

\begin{proof}
We invoke the equivalence of weak and strong PAC-learning \cite{DBLP:journals/ml/Schapire90, domingo2000madaboost} to conclude the desired expression. Note that, importantly, our weak distributional PAC learner works for all $\rho$ after taking the probability over $f \sim \mu$ and the randomness of the learning algorithm $A$. If the quantifiers were in another order, then we could not guarantee boosting since there would be no guarantee that the same set of functions of $\fC$, understood as a subset of $\{0,1\}^{s(n)}$, could be learned.
\end{proof}

\subsection{DistPAC-Learning vs HeurPAC-Learning}

\begin{theorem}
Let $\eval$ be an evaluation rule, and let $\fC$ be the $\eval$-induced $s(n)$-represented concept class.
Let $U$ be the uniform distribution over $\{0,1\}^{s(n)}$.
If $(\fC, U)$ is distPAC-learnable in time $t$, then $\fC$ is heurPAC-learnable in time $t$.
\end{theorem}
\begin{proof}
By definition, we can take the distPAC learner $A$ for $(\fC, U)$ and apply it as a heurPAC learner for $\fC$. To see this, observe that $A$ satisfies the heurPAC learning guarantee because it already obtains PAC-learning (i.e., learning with respect to \textit{any} example distribution) for a $1-\eta$ fraction of $\fC$. For heurPAC-learning, we in fact only need that for each example distribution, a $1-\eta$ fraction of $\fC$ is learnable.
\end{proof}

\subsection{Distribution-Specific Learning}

To obtain cryptographic primitives from the hardness of learning, the literature often considers hardness of \textit{distribution-specific} learning (as in \cite{blum1993cryptographic,nanashima2021theory}). In distribution-specific distPAC-learning, we intentionally fix the example distribution $\rho$. This limits the robustness of distPAC-learning, and entirely collapses it into Nanashima's notion of distribution-specific heurPAC-learning.

\begin{definition}[Distribution-specific distPAC-learning]\label{def:DS_distPAC}
    Let $\eval$ be an evaluation rule, and let the $\eval$-induced concept class $\fC$ be $s(n)$-represented. The pair $(\fC, \mu)$ is distributionally PAC-learnable over the example distribution $\rho$ if there exists an algorithm $A$ such that, for any $n \in \nat, \epsilon, \delta, \eta > 0$,
\begin{equation}
    \Pr_{f \sim \mu}\left[\Pr_{A}\left[\Pr_{x \sim \rho}\left[h(x) \not= f(x) : h \leftarrow A^{\ExO{f}{\rho}}(n, \epsilon, \delta, \eta)\right] \le \epsilon \right] \ge 1-\delta \right] \ge 1-\eta
\end{equation}
When $A$ runs in time $\poly(n, s(n), \epsilon^{-1}, \delta^{-1}, \eta^{-1})$, we say that $(\fC, \mu)$ is efficiently distPAC-learnable over $\rho$. 
\end{definition}

\subsubsection{Cryptography from Hardness of Distribution-Specific DistPAC-Learning}

The collapse of distPAC-learning and heurPAC-learning in the distribution-specific setting is a feature of distPAC-learning:
using results of \cite{nanashima2021theory} on one-way functions from hardness of distribution-specific heurPAC-learning, we obtain the same in distribution-specific distPAC-learning.

\begin{theorem}[OWFs from hardness of $\rho$-specific distPAC-learning]
Suppose that ${\sf P/poly}$ is not efficiently distPAC-learnable with respect to the uniform distribution over concept representations, on a polynomial time samplable example distribution $\rho$. Then, there exists an (infinitely-often) one-way function.
\end{theorem}
\begin{proof}
Let $U$ be the uniform distribution over $\{0,1\}^{\poly(n)}$.
Observe that $({\sf P/poly}, U)$ is efficiently distPAC-learnable on $\rho$ if and only if ${\sf P/poly}$ is efficiently heurPAC-learnable on $\rho$. Now, the claim follows immediately from Corollary 9 of \cite{nanashima2021theory}.
\end{proof}

\subsection{Hardness of DistPAC-Learning from Cryptography}\label{sec:hard_distPAC_crypto}

In the other direction, we can use the classic results \cite{goldreich1986construct, haastad1999pseudorandom} to easily see that the existence of a one-way function implies that there exists a polynomial time samplable $\mu$ such that $({\sf P/poly}, \mu)$ is not efficiently distPAC-learnable on any polynomial time samplable $\rho$.

Additionally, we now demonstrate that the technique for proving hardness of PAC-learning in Valiant's model (by Kearns and Valiant \cite{kearns1994cryptographic}), also works for distPAC-learning. The technique gives concrete ``hard target distributions'' for distPAC-learners. Note, it does not imply that distPAC-learning is hard \textit{for every} target distribution.

We sketch the proof technique of Kearns and Valiant, but direct the reader to \cite{kearns1994cryptographic} for the technicalities and also \cite{klivans2009cryptographic} for a more in-depth overview. 
The proof technique is built around the idea that learning the decryption function of a public key cryptosystem should be hard. Indeed, this is true by the following reasoning. For any given public key cryptosystem, an attacker can simulate an example oracle for the decryption function, by sampling a public and private key pair, and then generating encryptions of ones and zeros (with equal probability), using the randomized encryption function. Then, if there is a PAC-learning algorithm (in Valiant's model) for the decryption function, the attacker can apply it to the dataset collected as described, and then predict the messages of new ciphertexts, which breaks security.

We show that it suffices to have a distPAC-learning algorithm, rather than a PAC-learning algorithm in Valiant's model, to break security of a public key cryptosystem in this way.

\begin{theorem}\label{thm:KVtech}
    Suppose there exists a secure public key encryption scheme $({\sf gen, enc, dec})$ that generates $n$-bit encryptions of single bit messages, where decryption error is $\epsilon:= \epsilon(n)\le 1/n^{\omega(1)}$. Let $\mu$ be the distribution over decryption functions ${\sf dec}({\sf sk}, \cdot )$ with ${\sf sk, pk} \sim {\sf gen}$ and hard-wired. If for every ${\sf sk}$, ${\sf dec}({\sf sk}, \cdot ) \in \Lambda$, then $(\Lambda, \mu)$ is not efficiently distPAC-learnable.
\end{theorem}

\begin{proof}
    Suppose towards a contradiction that $(\Lambda, \mu)$ is efficiently distPAC-learnable using an algorithm $A$. We will show that there is a probabilistic polynomial time algorithm $D$ that breaks the security of $({\sf gen, enc, dec})$, by proving the following distinguishing equation:
    \begin{equation}\label{eq:dist}
    {\sf adv}(D) := \left|\Pr[D({\sf pk}, {\sf enc}({\sf pk}, 1)) =1 ] - \Pr[D({\sf pk}, {\sf enc}({\sf pk}, 0)) =1] \right| \ge 1/\poly(n)
    \end{equation}
    Here, the probabilities are taken over the internal randomness of $D$ and the encryption function, as well as ${\sf sk, pk} \sim {\sf gen}$.
    $D({\sf pk}, z)$ works as follows. 
    \begin{enumerate}
        \item Prepare a simulated oracle $\ExO{{\sf dec}({\sf sk}, \cdot)}{\rho}$, where each example is created by first choosing $y \in \{0,1\}$ uniformly randomly, and then encrypting $x = {\sf enc}({\sf pk}, y)$, and taking $\langle x, y\rangle$ as the example. $\rho$ denotes the marginal distribution over examples $x$.
        \item Execute $h \gets A^{\ExO{{\sf dec}({\sf sk}, \cdot)}{\rho}}(n, \frac{1}{10}, \frac{1}{10}, \frac{1}{10})$.
        \item Output $h(z)$.
    \end{enumerate}
    We prove \eqref{eq:dist} holds. Since $A$ is an efficient distPAC-learner, it only requires at most $\poly(n)$ examples. Therefore, since decryption error is $\epsilon:= 1/n^{\omega(1)}$, a union bound indicates that with probability $1- 1/n^{\omega(1)}$, no decryptions occur. Assume this is the case. By the assumption that $A$ is an efficient distPAC-learning algorithm for $(\Lambda, \mu)$, we have that, over the randomness of $h \gets A^{\ExO{{\sf dec}({\sf sk}, \cdot)}{\rho}}(n, \frac{1}{10}, \frac{1}{10}, \frac{1}{10})$, and the randomness of ${\sf dec}({\sf sk}, \cdot) \sim \mu$ and $x \sim \rho$,
    \begin{equation}
        \Pr\left[h(x) = {\sf dec}({\sf sk}, x)\right] \ge \frac{1}{2} + \frac{1}{\poly(n)}
    \end{equation}
Therefore, 
\begin{align*}
    {\sf adv}(D) &= \left|\Pr_{D,{\sf gen}, {\sf enc}}[h({\sf enc}({\sf pk},1)) =1] - \Pr_{D,{\sf gen}, {\sf enc}}[h({\sf enc}({\sf pk},0)) = 1\right|
    = 2  \Pr_{\substack{D,{\sf gen}, {\sf enc},\\ y \sim \{0,1\}}}[h({\sf enc}({\sf pk},y)) = y] - 1\\
    &= 2  \left(\Pr_{\substack{D,{\sf gen}, {\sf enc},\\ y \sim \{0,1\}}}[h({\sf enc}({\sf pk},y)) = {\sf dec}({\sf sk},({\sf enc}({\sf pk},y)))] - \Pr_{\substack{D,{\sf gen}, {\sf enc},\\ y \sim \{0,1\}}}[{\sf dec}({\sf sk},({\sf enc}({\sf pk},y))) \not= y]\right) - 1\\
    &=2  \left(\Pr_{\substack{D,{\sf gen}, {\sf enc}, \\x \sim \rho}}[h(x) = {\sf dec}({\sf sk},x)] - \epsilon\right)- 1\\
    &\ge \frac{1}{\poly(n)}
\end{align*}

Since we assumed that no decryption error occurred, we need to consider the case that it does. We already know that probability of any decryption error is $1/n^{\omega(1)}$. Hence, by another union bound, we get that ${\sf adv}(D) \ge 1/\poly(n) - 1/n^{\omega(1)} \ge 1/\poly(n)$.
\end{proof}


\section{No General Implication from Natural Proofs to PAC-Learning}\label{section:NisanNatProp}

In this section, we show that for circuit classes that can compute $\majthr$-circuits, there can be no general implication from natural proofs to PAC-learning algorithms in Valiant's model, unless $\tilde{O}(n^{1.5})$-$\mathrm{uSVP}$ has a polynomial time solution, and $\tilde{O}(n^{1.5})$-$\mathrm{SVP}$ and $\tilde{O}(n^{1.5})$-$\mathrm{SIVP}$ have polynomial time quantum solutions. We sketch an analogous fact for natural proofs for DNFs, and a natural assumption on the polynomial time hardness of refuting random $k$-SAT instances.
Additionally, we will explain that there exists a special target distribution such that there cannot even be an implication from natural proofs to distPAC-learning algorithms for that distribution, under the same lattice assumptions. 

\subsection{Nisan's Natural Proof}

We will formally define the natural property that arises from Nisan's natural proofs.
We begin by defining what is a natural property. Let $F_n$ be the set of all Boolean functions on $n$ inputs. Typically, the Boolean functions in the context of circuit complexity are $f: \{0,1\}^n \rightarrow \{0,1\}$. We will continue to use $f: \{0,1\}^n \rightarrow \{-1,1\}$ to maintain continuity with the previous sections.

\begin{definition}[Natural Property \cite{razborov1997natural}]
    A Natural Property is a sequence of subsets $Q = \{Q_n\}_{n \in \nat}$ of $F = \{F_n\}_{n \in \nat}$, if it satisfies the following conditions.
    \begin{enumerate}
        \item \textbf{Constructivity.} The predicate ``is $f$ contained in $Q_n$'' can be computed in polynomial time.
        \item \textbf{Largeness.} $|Q_n| \ge \delta_n \cdot |F_n|$.
        \item \textbf{Usefulness.} For any sequence of functions $f_n \in F_n$ (for $n \in \nat$), if $f_n \in \Lambda$ then $f_n \notin Q_n$ almost everywhere.
    \end{enumerate}
    When $Q_n$ satisfies these conditions we say that it is a natural property for $\Lambda$ with density $\delta_n$. In general, we fix $\delta_n = 1/2$.\footnote{See Lemma 2.7 of \cite{carmosino2016learning} for an explanation for why this is reasonable.}
\end{definition}

Now we will define the property that arises from Nisan's lower bounds for $\majthr$-circuits \cite{nisan1993communication}, and then demonstrate that it is natural. To the best of our knowledge, the natural property has never been explicitly formalized prior to this paper, but the fact that it is natural is essentially credited to the sequence of works \cite{nisan1993communication, chung1993communication, raz2000bns, viola2007norms}; \cite{nisan1993communication} provided the lower bound, and \cite{chung1993communication, raz2000bns, viola2007norms} all implicitly proved it was natural.
The property is defined by the following algorithm that classifies whether or not a truth table is in the property.

\begin{definition}[Nisan's Natural Property]\label{def:natprop}
\color{white}.
\color{black}
\begin{enumerate}
    \item \textbf{Input.} $T_{f_n} \in \{0,1\}^{2^n}$ a truth table of a function $f_n \in F_n$.
    \item Choose any partition of the $n$ inputs of $f_n$ into two sets $A,B$ with $|A| = \lceil n/2 \rceil$ and $|B| = \lfloor n/2 \rfloor$. Let $x_1||x_2$ denote the partitioned input according to $A,B$.
    \item Given $T_{f_n}$, compute \[\alpha = 2^n \cdot R_2(f_n) = \sum\limits_{\substack{x^0_1,x^1_1 \in \{0,1\}^{\lceil n/2 \rceil}\\ x^0_2,x^1_2 \in \{0,1\}^{\lfloor n/2 \rfloor}}}{\prod_{b_1,b_2 \in \{0,1\}}f_n(x^{b_1}_1|| x^{b_2}_2)}\]
    \item \textbf{Output.} If $|\alpha| \le 2^{2n/3}$, print $1$, otherwise print $0$. 
\end{enumerate}
\end{definition}

Let $\comp{\maj_s}{\thr}$ denote the class of majority-of-thresholds circuits where the top majority gate has fan-in $s$.

\begin{theorem}[\cite{nisan1993communication, chung1993communication, raz2000bns, viola2007norms}]\label{MajThrNatProp}
There is a constant $c$ such that there is a natural property for $\comp{\maj_{2^{n/c}}}{\thr}$ with density 1/2.
\end{theorem}

To prove this theorem, one uses the result of Nisan \cite{nisan1993communication}.
The notation $\RANDCC{2}{c}{\gamma}$ denotes the set of functions $f: \{0,1\}^n \rightarrow \{-1,1\}$ that have 2-party randomized communication complexity, with cost at most $c$ and bias at least $\gamma$, for every partition of the inputs to the function. See \secref{prelims} for formal definitions of communication complexity classes and protocols in this work.

\begin{theorem}[\cite{nisan1993communication}]\label{nisanThm}
$\comp{\maj_s}{\thr} \subseteq \RANDCC{2}{O(\log(s))}{O(1/s)}$
\end{theorem}

\begin{proof}[Sketch.]
    The statement follows from combining a few arguments in \cite{nisan1993communication}. First, Nisan proves that a single $\thr$ gate is contained in $\RANDCC{2}{O(\log(s))}{1/2- O(1/s)}$ (Theorem 1a.). 
    Second, Nisan shows that the majority vote of any $s$ functions contained in $\RANDCC{2}{O(\log(s))}{1/2- O(1/s)}$ is contained in $\RANDCC{2}{O(\log(s))}{O(1/s)}$ (Lemma 5). 
    Therefore we get that $\comp{\maj_s}{\thr} \in \RANDCC{2}{O(\log(s))}{O(1/s)}$.
\end{proof}

Now we prove \thmref{MajThrNatProp}.

\begin{proof}[Proof of \thmref{MajThrNatProp}]
    We prove each of the three necessary properties individually.
    First, observe that \defref{def:natprop} defines a polynomial time algorithm, in the size of the truth table (this proves the Constructivity property).
    
    For the Largeness property, we will analyze the probability that the algorithm outputs 1 given a random truth table as input. 

    Consider $\alpha/2^{-n} = R_2(f_n)$. We have that 
    \[
        \alpha/2^{-n} = \Ex{x^0_1,x^1_1,x^0_2,x^1_2}{{\prod_{b_1,b_2 \in \{0,1\}}f_n(x^{b_1}_1, x^{b_2}_2)}}
    \]
    This expected value over the uniformly random choice of $x^0_1,x^1_1,x^0_2,x^1_2$ is 0 except for when either $x^0_1 = x^0_2$ or $x^1_1 = x^1_2$. By a union bound, we have then that $\alpha/2^{-n} \le 2^{1-(n+1)/2}$. Therefore, $\alpha \le 2^{n/2+1}$, and the probability over a random truth table $T_{f_n}$ that $\alpha \le 2^{2n/3}$ is at least 1/2. This proves density is at least 1/2.

    Finally, we prove Usefulness. By \thmref{thm:corrBound} and \thmref{nisanThm}, we know that when $T_{f_n}$ is the truth table of $f_n \in \comp{\maj_s}{\thr}$, then for $s \le 2^{n/c}$ and some sufficiently large constant $c$, we have that $\alpha \ge 2^n/\poly(2^{n/c}) \ge 2^{2n/3}+1$. Thus, whenever $f_n \in \comp{\maj_{2^{n/c}}}{\thr}$, the algorithm defining the natural property outputs 0, as desired.

\end{proof}

\subsection{Proof of \thmref{intro:natBarrier}}

In the last section, we demonstrated that $\comp{\maj_s}{\thr}$ has a dense natural property for $s$ up to $2^{n/c}$ for some constant $c$. Now we will use this fact, together with hardness of learning $\comp{\maj}{\thr}$ in Valiant's PAC model \cite{klivans2009cryptographic} to prove \thmref{intro:natBarrier}.

\begin{theorem}[Theorem 1.3 in \cite{klivans2009cryptographic}]\label{thm:KS}
    Assume that $\comp{\maj}{\thr}$ is PAC-learnable in time $\poly(n)$. Then there is a polynomial time solution to $\tilde{O}(n^{1.5})$-$\mathrm{uSVP}$, and polynomial time quantum solutions to $\tilde{O}(n^{1.5})$-$\mathrm{SVP}$ and $\tilde{O}(n^{1.5})$-$\mathrm{SIVP}$.
\end{theorem}

\begin{theorem}[\thmref{intro:natBarrier} restated]\label{thm:natBarrier}
    Let $u: \nat \rightarrow \nat$. Suppose that a natural proof against $\comp{\maj_{u(n)}}{\thr}$-circuits (and density 1/2) implies that $\poly(n)$ size $\comp{\maj}{\thr}$-circuits are PAC-learnable in Valiant's model, in time $\exp(u^{-1}(\poly(n)))$. Then, then there is a polynomial time solution to $\tilde{O}(n^{1.5})$-$\mathrm{uSVP}$, and polynomial time quantum solutions to $\tilde{O}(n^{1.5})$-$\mathrm{SVP}$ and $\tilde{O}(n^{1.5})$-$\mathrm{SIVP}$.
\end{theorem}

\begin{proof}
By \thmref{MajThrNatProp}, there is a constant $c$ such that there is a natural property for $\comp{\maj_{2^{n/c}}}{\thr}$ with density 1/2. Thus, by the condition of the current theorem, we conclude that the class of $\comp{\maj}{\thr}$-circuits is PAC-learnable in Valiant's model.
By \thmref{thm:KS}, we then obtain a polynomial time solution to $\tilde{O}(n^{1.5})$-$\mathrm{uSVP}$, and polynomial time quantum solutions to $\tilde{O}(n^{1.5})$-$\mathrm{SVP}$ and $\tilde{O}(n^{1.5})$-$\mathrm{SIVP}$.     
\end{proof}

\subsubsection{Hardness of DistPAC-Learning for Majority of Threshold Circuits for all Target Distributions}

\begin{corollary}\label{cor:psampHardMajThr}
    There exists a polynomial time samplable distribution $\mu$ over polynomial size $\majthr$-circuits, such that if $(\majthr, \mu)$ is efficiently distPAC-learnable, then there is a polynomial time solution to $\tilde{O}(n^{1.5})$-$\mathrm{uSVP}$, and polynomial time quantum solutions to $\tilde{O}(n^{1.5})$-$\mathrm{SVP}$ and $\tilde{O}(n^{1.5})$-$\mathrm{SIVP}$.     
\end{corollary}

\begin{proof}
    Implicitly from \thmref{thm:KS}, we know that the decryption function of Regev's public key cryptosystem can be implemented by a $\majthr$-circuit, with negligible decryption error. Therefore, by \thmref{thm:KVtech}, the statement follows.
\end{proof}

\subsubsection{A Similar Barrier via Hardness of Learning DNFs}

A similar statement to \thmref{thm:natBarrier} can be observed for DNFs using a different hardness of PAC-learning result of \cite{daniely2016complexity}. Specifically, \cite{daniely2016complexity} prove that PAC-learning DNFs in Valiant's model implies the ability to refute random $k$-SAT instances. Feige \cite{feige2002relations} first introduce hardness assumptions regarding refuting $k$-SAT instances.

An algorithm refutes random $k$-SAT instances with $c(n) \ge \Omega(n)$ clauses if on $1 - o(1)$ fraction of the $k$-SAT formulas with $c(n)$ constraints, it outputs
“unsatisfiable”, while whenever it encounters a satisfiable $k$-SAT formula with $c(n)$ constraints, it outputs “satisfiable.” A random K-SAT instance $I = \{C_1,\cdots , C_{c(n)}\}$ is sampled in the way that each clause/constraint $C_i$ is chosen uniformly at random from the set of $k$-SAT clauses/constraints with $n$ variables.

Thus, it follows directly from the assumption regarding hardness of refuting random $k$-SAT formulas from \cite{daniely2016complexity} that there is no implication from a Natural Proof for exponential size DNFs, to PAC-learning for DNFs. This is because Nisan's natural proof clearly works for exponential size DNFs, since the set is clearly contained by the set of exponential size majority-of-threshold circuits.

\section{Main Theorem}\label{section:mainTHM}

In this section, we will prove \thmref{intro:mainResult}, which obtains distPAC-learning algorithms for concept classes that have low-cost associated communication games. To do so, we will first define the communication game, and then obtain a weak distPAC-learning algorithm (see \defref{def:weakDistPAC}). Finally, we will conclude \thmref{intro:mainResult} using the equivalence between weak and strong distPAC-learning (see \thmref{distPAC_weak_strong_equiv}).

In the following sections, we will use \thmref{intro:mainResult} to derive distPAC-learning algorithms for natural distributions over $\majthr$-circuits, polytopes, and DNFs, and impossibility results for weak PRFs that can be evaluated with $\majthr$-circuits with an arbitrary input encoding.

\subsection{Communication Games}

Recall that $\fC = \{\fC_n\}_{n \in \nat}$ is the $s(n)$-represented Boolean concept class that is induced by the evaluation function $\eval = \{\eval_n\}_{n \in \nat}$. We define the communication game associated with $\fC$:

\begin{definition}[2-party distributional communication game]
With respect to a evaluation rule $\eval$ and product distribution $(\mu, \rho)$, the 2-party communication game $\fG[\eval, n, (\mu, \rho)]$ is the following:
\begin{itemize}
    \item Setup: $\pi_f \sim \mu, x \sim \rho$.
    \item Player 1 gets as input $\pi_f \in \{0,1\}^{s(n)}$ for concept $f \in \fC_n$.
    \item Player 2 gets as input a string $x \in \{0,1\}^n$.
    \item The object of the game is for the parties to output the value $\eval_n(\pi_f, x) = f(x)$, using as few bits of communication as possible.
\end{itemize}

We say that $\fG[\eval, n, (\mu, \rho)]$ is $(c, \gamma)$-evaluated if the parties can communicate at most $c$ bits, and win the game with probability $1/2+\gamma$ (over the random sample of inputs according to $(\mu, \rho)$).
\end{definition}
\paragraph{Other definitions.} We direct the reader to \secref{prelims} for the necessary definitions of communication complexity.

\subsection{Weak Learning}\label{sec:weakLearning}

Towards \thmref{intro:mainResult}, we will start by first obtaining a weak learning algorithm, which only requires prediction accuracy marginally better than a coin toss. 

\paragraph{Notation.} In the following, we discuss boolean functions $f: \{0,1\}^n \rightarrow \{-1,1\}$, and denote by $U_n$ the uniform distribution over $\{0,1\}^n$. 
For shorthand, we will write $c:= c(n), \gamma:= \gamma(n)$, to denote number of bits of communication and protocol bias, which are dependent on $n$, the input length of a concept. 

Also, in the rest of the paper we will streamline notation by eliding the subscripts on distributions coming from ensembles indexed by $n \in \nat$.
\bigskip 

\begin{theorem}\label{weakMainResult}
    Let $\eval$ be an evaluation rule.
    Suppose that, for every $n \in \nat$, and product distribution $(\mu, \rho)$, $\fG[\eval, n, (\mu, \rho)]$ is $(c, \gamma)$-evaluated. 
    Then there exists an algorithm $A$ such that, for any $n \in \nat, \delta, \eta > 0$,
     \begin{equation*}
        \Pr_{f \sim \mu}\left[\Pr_{A} \left[\forall \rho: \Pr_{x \sim \rho}\left[h(x) \not= f(x) : h \leftarrow A^{\ExO{f}{\rho}}(n, \delta, \eta)\right] \le \frac{1}{2}-\poly(\gamma \cdot 2^{-c})  \right] \ge 1-\delta \right] \ge 1-\eta
    \end{equation*}
For $\mu$ samplable in time $t(n)$, $A$ runs in time $\poly(n, t(n), s(n), \gamma^{-1}, \delta^{-1}, 2^c)$.

\end{theorem}

\begin{proof}

To construct $A$, we will follow three steps:
\begin{enumerate}
    \item Construct a weak randomized predictor $L$.
    \item Argue that many good non-uniform but deterministic predictors
    exist, by fixing coins and samples for
    $L$.
    \item Construct a deterministic predictor by sampling and then testing
    enough non-uniform predictors.
\end{enumerate}
Steps 2 and 3 follow from standard techniques (i.e., ``constructive averaging'').
 
\begin{claim}[Weak randomized predictor]\label{randPredictor}
Let $\eval$ be an evaluation rule. Under the conditions of \thmref{weakMainResult}, there exists a randomized algorithm $L$, running in time $\poly(n, t(n), s(n), \gamma^{-1}, \delta^{-1},\eta^{-1}, 2^c)$, such that for any $n \in \nat, \delta, \eta >0$, the following equation is satisfied:
\begin{equation}\label{weakPredictor}
        \Pr_{f \sim \mu}\left[\Pr_{L}\left[\forall \rho: \Pr_{z \sim \rho}\left[L^{\ExO{f}{\rho}}(z, n, \delta, \eta) \not= f(z) \right] \le \frac{1}{2}-\poly(\gamma \cdot 2^{-c})  \right] \ge 1-\delta \right] \ge 1-\eta
    \end{equation}
\end{claim}

\begin{proof}[Proof of Claim \ref{randPredictor}]

We will abuse notation and write $\pi_g \sim \mu$ to denote the binary representation of a concept $g$, distributed appropriately according to the target distribution $\mu$ over the $\eval$-induced $s(n)$-represented concept class $\fC$.
See the randomized predictor $L$ in Figure \ref{fig:predictor}.

Consider the distribution $\mathcal{M}$ over $2 \times 2$ matrices 
    \[
    C = 
\begin{bNiceMatrix}
        \eval(\pi_f,z) & \eval(\pi_f,w)\\
        \eval(\pi_g,z)& \eval(\pi_g,w)\\
\end{bNiceMatrix}
\]
where $\pi_f, \pi_g \sim \mu$ and $z,w \sim \rho$. We now claim that, under the conditions of \thmref{weakMainResult}, this distribution is efficiently \textit{distinguishable} from the distribution $\mathcal{R}$ over random $2 \times 2$ matrices,
    \[
    R = 
\begin{bNiceMatrix}
        r_{00} & r_{01}\\
        r_{10}& r_{11}\\
\end{bNiceMatrix}
\]
To see this, observe that the distribution over $C$ is identical to the following distribution over $2 \times 2$ matrices (abusing notation, $z =\rho(y)$ denotes a point $z$ sampled according to $\rho$ with the random bits $y$)
\[
    D =  
\begin{bNiceMatrix}
        \eval(\mu(x),\rho(y)) & \eval(\mu(x),\rho(y'))\\
        \eval(\mu(x'),\rho(y))& \eval(\mu(x'),\rho(y'))\\
\end{bNiceMatrix}
\]
Here, $x,x',y,y'$ are uniformly random strings. We can assume that without loss of generality that $x,x',y,y'$ are all the same length, by defaulting to the maximum necessary length for sampling $\mu, \rho$ (and padding the shorter strings with useless random bits). Therefore, identifying $\xi(x,y) := \eval(\mu(x),\rho(y))$, we can now see that 
\[
R_2(\xi) = \mathop{\mathbb{E}}_{\substack{\pi_f,\pi_g\\z,w}}\left[\prod_{i,j \in \{0,1\}} C_{ij}\right]
\]
It now readily follows that when 
$\fG[\eval,n, (\mu, \rho)]$ is $(c, \gamma)$-evaluated (which is true by assumption), then 
\begin{equation}\label{eq:avgPAC_dist}
R_2(\xi) = \mathop{\mathbb{E}}_{\substack{\pi_f,\pi_g\\z,w}}\left[\prod_{i,j \in \{0,1\}} C_{ij}\right] \ge (\gamma \cdot 2^{-c})^4
\end{equation}

\begin{figure}
    \centering
\begin{algorithm}[H]
  \caption{$L^{\ExO{f}{\rho}}$}
  \label{L1}
  \begin{algorithmic}[1]
    \State \textbf{Input}: $z \in \{0,1\}^n, n \in \nat, \delta, \eta>0$
    \State Pick uniformly random values $b_1,b_2 \in \{0,1\}$.
    \State Pick uniformly random values $r_{00},r_{01},r_{10},r_{11} \in \{-1,1\}$
    \State Sample $\pi_g \sim \mu$.
    \State Sample $(w,y) \sim \ExO{f}{\rho}$.
    \If{$b_1 = b_2 = 0$}
    \State $v \gets \prod_{i,j \in \{0,1\}}r_{ij}$
    \EndIf\smallskip
    \If{$b_1 = 0, b_2 = 1$}
    \State $v \gets y \cdot r_{00} \cdot \prod_{i,j \in \{0,1\}}r_{ij}$
    \EndIf\smallskip
    \If{$b_1 = 1, b_2 = 0$}
    \State $v \gets \eval(\pi_g,w) \cdot \eval(\pi_g,z) \cdot \prod_{j \in \{0,1\}}r_{1j}$
    \EndIf\smallskip
    \If{$b_1 = 1, b_2 = 1$}
    \State $v \gets y \cdot \eval(\pi_g,w) \cdot \eval(\pi_g,z) \cdot r_{11}$
    \EndIf\smallskip
    \State $b \gets r_{b_1b_2}$
    \State \textbf{Output} $b \cdot v$
  \end{algorithmic}
\end{algorithm}
\vspace{-0.5cm}

\caption{Randomized predictor $L$. }
    \label{fig:predictor}
\end{figure}

On the other hand, 
\[
\mathop{\mathbb{E}}_{\substack{R}}\left[\prod_{i,j \in \{0,1\}} R_{ij}\right] =0
\]

Now that we have established this, we may proceed by a hybrid argument. Define the neighboring hybrid distributions $H_{1},H_{2}, H_{3},H_{4}, H_{5}$ over $2 \times 2$ matrices, as in Figure \ref{fig:Hybs}.
\begin{figure}
    \centering
\begin{align*}
    H_1 &= \mathcal{R}\\
    H_2 &= 
\begin{cases}
C_{k\ell} \text{ when $k = 0, \ell= 0$} \\
R_{k\ell} \text{ otherwise}
\end{cases}\\
H_3 &= 
\begin{cases}
C_{k\ell} \text{ when $k = 0, \ell \le 1$} \\
R_{k\ell} \text{ otherwise}
\end{cases}\\
H_4 &= 
\begin{cases}
C_{k\ell} \text{ when $k = 0$ or $\ell \le 0$} \\
R_{k\ell} \text{ otherwise}
\end{cases}\\
H_5 &= \mathcal{M}
\end{align*}
\vspace{-0.5cm}
\caption{Hybrid sequence.}
    \label{fig:Hybs}
\end{figure}

It then follows that for random hybrid neighbors $H_{i}, H_{i+1}$ ($i\in [4]$),
\begin{align}\label{neighbor1}
    \mathop{\mathbb{E}}_{i \sim [4]}\left[\mathop{\mathbb{E}}_{H' \sim H_{i+1}}\left[\prod_{k,j \in \{0,1\}} H_{kj}' = 1\right] - \mathop{\mathbb{E}}_{H \sim H_i}\left[\prod_{k,j \in \{0,1\}} H_{kj} = 1\right]
    \right] \ge (\gamma \cdot 2^{-c})^4/4
\end{align}

To ease notation, let $D(H) =\prod_{k,j \in \{0,1\}} H_{kj}$, and let $V_i$ denote the event that $D(H_i)= 1$. Intuitively, the function $D$ stands for ``distinguisher,'' and can be thought of as such.

We continue by observing that, by definition, the value stored as $v$ in $L$ (\algoref{L1}) is $D(H_i)$ for a random $i \in [4]$. Hence, the output of $L$, which is written as $D(H_i) \cdot b$, is interpreted as a prediction, where $b = r_{b_1b_2}$ is the ``guess bit.'' Note that, the string $b_1b_2$ is the binary representation of $i$.

Now, conditioning on correctness of this guess bit,
we have that for all $\rho$, and probabilities taken over $z\sim \rho, f\sim \mu$ and the randomness of $L$:
\begin{align*}
    \Pr\left[L^{\ExO{f}{\rho}}(z, n, \delta, \eta) = f(z)\right]
    &= \Pr\left[L^{\ExO{f}{\rho}}(z, n, \delta, \eta) = f(z) \mid b = f(z)\right]
      \cdot \Pr[b = f(z)] \\
    &\ \ \ \ + \Pr\left[L^{\ExO{f}{\rho}}(z, n, \delta, \eta) = f(z)
      \mid b \not= f(z)\right] \cdot \Pr[b \not= f(z)] \\
    &= \frac{1}{2}\Big(\Pr[b \cdot D(H_{i}) = f(z) \mid b = f(z)]\\
    &\ \ \ \ +\Pr[b \cdot D(H_i) = f(z) \mid b \not= f(z)]\Big)\\
\end{align*}

Indeed, when $V_i$ is unsatisfied, this means that the output of $L$ is $b$. The case analysis follows:

\begin{align*}
    \Pr[
    L^{\ExO{f}{\rho}}(z, n, \delta, \eta)
    = f(z)] &=
              \frac{1}{2}\Big(\Pr[V_i \mid b = f(z)]
              + \Pr[\lnot V_i \mid b \not= f(z)]\Big)\\
            &= \frac{1}{2}
              + \frac{1}{2}\Big(\Pr[V_i \mid b = f(z)]
              - \Pr[V_i \mid b \not= f(z)]\Big)\\
  \end{align*}
  By conditioning, we know that:
  \begin{equation*}
    \Pr[V_i] =
    \frac{1}{2} \Pr[V_i \mid b = f(z)]
    + \frac{1}{2}\Pr[V_i \mid b \not= f(z)]
  \end{equation*}
  rearranging the terms, we get:
  \begin{equation*}
    \frac{1}{2}\Pr[V_i \mid b \not= f(z)] =
    \Pr[V_i]
    - \frac{1}{2} \Pr[V_i \mid b = f(z)]
  \end{equation*}
  We thus conclude:
  \begin{equation*}
    \Pr[L^{\ExO{f}{\rho}}(z, n, \delta, \eta) = f(z)] =
    \frac{1}{2}
    + \underbrace{\Pr[V_i \mid b = f(z)]}_{(\alpha)}
    - \underbrace{\Pr[V_i]}_{(\beta)}
  \end{equation*}
The term $(\alpha)$ corresponds to the case that $L$ computes the 2-party norm on a sample from $H_{i+i}$ (i.e., the product of the entries of a matrix sampled from $H_{i+i}$), while term $(\beta)$ is the case that $L$ computes the 2-party norm on a sample from $H_{i}$ (the product of the entries of a matrix sampled from $H_{i}$). 
Thus, by equation \eqref{neighbor1},
  \begin{align*}\label{bound}
    \Pr[L^{\ExO{f}{\rho}}(z, n, \delta, \eta) = f(z)]
    &= \frac{1}{2} + (1-\Pr[D(H_{i+1}) = 1] - (1-\Pr[D(H_{i}) = 1])\\
    &\ge \frac{1}{2}+(\gamma \cdot 2^{-c})^4/8
  \end{align*}
\end{proof}

Having established Claim \ref{randPredictor}, we now
convert the randomized algorithm $L^{\ExO{f}{\rho}}$ into
a non-uniform learning algorithm by averaging. Let
$L^{\ExO{f}{\rho}}_{\mathrm{inp}, i}(z; r)$ denote the
Algorithm \ref{L1} where the random hybrid choice is fixed to be
$i$, and the input parameters $\mathrm{inp} = (n, \delta, \eta)$ are
hard-wired in, and the random bits $r$ for computing other randomized
aspects of the algorithm is treated as input. This allows us to consider the algorithm
as a deterministic mapping of random bits and examples from $\ExO{f}{\rho}$ to a
circuit that weakly agrees with $f$. By a standard averaging argument, we obtain:
\begin{claim}[Averaging, see lemma A.11 of \cite{arora2009computational}] 
  \label{claim:many-good-L-coins}
\[
\Pr_r\bigg[\Pr_{z \sim \rho}
    \Big[L_{\mathrm{inp}, i}^{\ExO{f}{\rho}}(z; r) = f(z)\ |\ r\Big]
    > \frac{1}{2}+\frac{(\gamma \cdot 2^{-c})^4}{32} \bigg] > (\gamma \cdot 2^{-c})^4/4
\]
\end{claim}
Taking hybrid index $i$ and $r$ uniformly at random, we obtain
``good'' choices with good probability. Therefore such a circuit is
efficiently found by randomized trial-and-error; we
sample many candidate predictors in parallel and then compare each to
the concept by checking random examples. By a
standard application of Chernoff bounds, sufficiently many examples
will be enough to check that a circuit with good enough accuracy is indeed good enough, with high probability.

\begin{claim}[Without proof]
  \label{claim:easy-pred-filtering-heur}
  With probability $1-\delta$, $A$ (Algorithm \ref{A} in Figure \ref{fig:sampling})
  outputs a ``good'' circuit that correctly classifies
  $1/2+ \poly(\gamma \cdot 2^{-c}) $ fraction of points, where $t,m$ are quantities that are polynomially bounded as a function of $\poly(\gamma \cdot 2^{-c})$ and $\log(\delta^{-1})$.
\end{claim}

From the above claims it now follows, from a Markov argument, that:
\begin{equation*}
    \Pr_{\pi_f \sim \mu}\left[\Pr_{A}\left[\forall \rho : \Pr_{z \sim \rho}\left[h(z) \not= f(z) : h \leftarrow A^{\ExO{f}{\rho}}(n, \delta)\right] \le \frac{1}{2}- \poly(\gamma \cdot 2^{-c}) \right] \ge 1-\delta/\eta \right] \ge 1-\eta
\end{equation*}
This concludes the proof of \thmref{weakMainResult}.

\begin{figure}
    \centering
\begin{algorithm}[H]
\setstretch{1}
    \caption{$A^{\ExO{f}{\rho}}$}\label{A}
    \begin{algorithmic}[1]
      \State \textbf{Input}: $n \in \nat, \delta \in (0,1]$
      \State Sample $m$ (sufficiently many) candidate circuits,
      using oracle access to $\ExO{f}{\rho}$ as needed.
      \State Sample $t$ (sufficiently many)
      additional random examples from $\ExO{f}{\rho}$.
      \For{each sampled circuit $C_i,$}
      \State Compute using random examples: $\alpha_i \gets \corr{C_i}$
      \EndFor
      \State \Output the circuit with largest $\alpha$ value.
    \end{algorithmic}
  \end{algorithm}
  \vspace{-0.5cm}
    \caption{Algorithm for sampling and testing candidate predictors.}
    \label{fig:sampling}
\end{figure}
\end{proof}

\paragraph{Remark.} Using the 2-party norm as we do is a \textit{universal} distinguisher. That is, it distinguishes any evaluation rule that is $(c, \gamma)$-evaluated from a random function (using the lower bound of $(\gamma \cdot 2^{-c})^4$). Therefore it holds that for \textit{arbitrary} choice of distribution $\rho$, we obtain the desired guarantee. Indeed, the choice of $\rho$ can be adversarial with respect to $f \sim \mu$, and it never needs to be known by $A$.

\subsection{Main Theorem}
\thmref{weakMainResult} is enough to prove \thmref{intro:mainResult}.
Recall that $\fC = \{\fC_n\}_{n \in \nat}$ is a boolean concept class that is $s(n)$-representable by the evaluation function $\eval = \{\eval_n\}_{n \in \nat}$.

\begin{theorem}[\thmref{intro:mainResult}, restated]\label{mainResult}
    Let $\eval$ be an evaluation rule, and suppose that for every $n \in \nat$,
    and product distribution $(\mu, \rho)$, $\fG[\eval, n, (\mu, \rho)]$ is $(c, \gamma)$-evaluated. 
    Then, $(\fC, \mu)$ is distPAC-learnable. For a time $t(n)$-samplable $\mu$ and $s(n)$-represented $\fC$, the learning algorithm runs in time 
    polynomial in $n, t(n), s(n), \gamma^{-1}, \epsilon^{-1}, \delta^{-1}, \eta^{-1}$, and $2^{c}$.
\end{theorem}
\begin{proof}
    Immediate from \thmref{distPAC_weak_strong_equiv} and \thmref{weakMainResult}.
\end{proof}

\paragraph{Remark.} The exponential dependency of $c$ is possibly necessary (i.e., necessary assuming exponentially secure one-way functions exist), since the theorem does not restrict $\fC$ (e.g., it can be ${\sf P/poly}$), and no matter what $c \le n$. Also, see \secref{sec:RequiresBreakthroughs} for further discussion.

\section{Distributional PAC-Learning from Nisan's Natural Proofs}\label{section:distPAC_algos}

In this section, we will apply \thmref{mainResult} to obtain polynomial time distPAC-learning algorithms. \thmref{thm:distPAC_MajThr} restates \thmref{intro:distPAC_MajThr}.

\begin{theorem}[distPAC-learning --- \thmref{intro:distPAC_MajThr} restated]\label{thm:distPAC_MajThr}
    Let $\eval \in \majthr$ be any evaluation rule, and let $\mu$ be any polynomial time samplable target distribution.
    Then, for the $\eval$-induced $s(n)$-represented concept class $\fC$, the pair $(\fC, \mu)$ is efficiently distPAC-learnable.
\end{theorem}

\begin{proof}[Proof of \thmref{thm:distPAC_MajThr}]

By \thmref{nisanThm}, $\majthr$ has a randomized communication protocol with cost $O(\log n)$ and bias at least $1/\poly(n)$. This implies that for any product distribution $(\mu, \rho)$ it is the case that $\fG[\eval, n, (\mu, \rho)]$ is $(O(\log n), 1/\poly(n))$-evaluated, since $\eval \in \majthr$. This is enough to conclude the theorem by \thmref{mainResult}.
\end{proof}

\subsection{DistPAC-Learning on Natural Target Distributions}

\subsubsection{Majority-of-Thresholds}\label{subsubsec:nat-dist-majthr}

In this section, we will use \thmref{thm:distPAC_MajThr} to show efficient distPAC-learning algorithms for $\majthr$-circuits over natural target distributions. Recall that, as mentioned in the introduction, we must restrict the target distribution over polynomial size $\majthr$-circuits if we hope to obtain efficient distPAC-learning without also obtaining polynomial time solutions to $\tilde{O}(n^{1.5})$-$\mathrm{uSVP}$, and polynomial time quantum solutions to $\tilde{O}(n^{1.5})$-$\mathrm{SVP}$ and $\tilde{O}(n^{1.5})$-$\mathrm{SIVP}$ (see \corref{cor:psampHardMajThr} for the formal statement).   

\paragraph{The target distribution.} Let $L = (T_1, \cdots T_{m})$ be a list of $m := \poly(n)$ linear threshold functions. Also, let $\mu$ be any $\poly(n)$ time samplable distribution over $\{0,1\}^{m}$. We define the distribution $\mu_L$ over $\majthr$-circuits as follows. 
\begin{itemize}
    \item Sample $\theta \sim \mu$.
    \item Output the $\majthr$-circuit that is the majority vote over each $T_i \in L$ such that $\theta_i = 0$.
\end{itemize}

We will show that, for any $\mu$, $L$ as described above, $(\majthr, \mu_L)$ is efficiently distPAC-learnable.

\begin{theorem}\label{distPAC_MajThr2}
    Let $L = (T_1, \cdots T_{m})$ be any list of $m := \poly(n)$ linear threshold functions over $n$-bit inputs, and let $\mu$ be any $\poly(n)$ time samplable distribution over $\{0,1\}^{m}$. The pair $(\majthr, \mu_L)$ is efficiently distPAC-learnable.
\end{theorem}

\begin{proof}
    To prove the theorem, we will show that the target distribution $\mu_L$ can be translated into a distribution $\mu_L^*$ over binary representations, so that given a representation $\pi_f \sim \mu_L^*$, there is an evaluation rule $\eval \in \majthr$ such that, for every $x \in \{0,1\}^n$, it holds that $\eval(\pi_f, x) =f(x)$. In other words, the distribution over functions $\eval(\pi_f, \cdot)$ for $\pi_f \sim \mu_L^*$ is identical (when considering functional equivalence) to $\mu_L$.
    This suffices to prove the theorem by invoking \thmref{mainResult}.

    We construct $\eval$ in the following way. For $T_i \in L = (T_1, \cdots T_m)$, let $\theta_i \in \ints$ be the threshold parameter and let $w_{i,j} \in \ints$ be the $j^{th}$ weight in $T_i$ (i.e., $T_i = [w_{i,1}x_1 + \cdots + w_{i,n}x_n \ge \theta_i]$). Now, for every $T_i \in L$, add an input variable $z_i$ that is weighted by $w_{z_i} = -(w_{i,1}+ \cdots + w_{i,n}+\theta_i+1)$. Call the new threshold function $T_i' = [w_{i,1}x_1 + \cdots + w_{i,n}x_n + w_{z_i}z_i\ge \theta_i]$
    
    Also, define auxiliary variables $y_1 \cdots y_m$.
    We define $\eval: \{0,1\}^{2m} \times \{0,1\}^n \rightarrow \{-1,1\}$
    \[
    \eval(z_1 \cdots z_m, y_1 \cdots y_{m}, x) := \maj (T_1'(x, z_1), \cdots, T_m'(x, z_m), y_1, \cdots, y_{m})
    \]
    We now define the distribution $\mu_L^*$ over $\{0,1\}^{2m}$. To sample $\mu_L^*$, first sample $z \sim \mu$. Then, letting $|z|$ denote the number of ones in $z$, sample $y = y_1\cdots y_m$ by choosing uniformly at random amongst the $m$-bit strings such that $|y| = (|z|+m)/2$. Finally, output $(z,y) \in \{0,1\}^{2m}$.

    It only remains to show that the distribution over functions $\eval(\pi_f, \cdot)$ for $\pi_f:= (z,y) \sim \mu_L^*$ is identical (when considering functional equivalence) to $\mu_L$. To see this, observe that whenever $z_i = 1$, $T'_i(x,z_i) = 0$. This follows from the fact that $-w_{z_i} > w_{i,1}+ \cdots + w_{i,n}+\theta_i$. On the other hand, clearly when $z_i = 0$, $T'_i(x,z_i) = T_i(x)$. 

    The goal is to output $\maj_i(T_i(x))$ for all $T_i$ such that $z_i = 0$. Therefore, we need to balance the presence of the extra zero votes in $\maj (T_1'(x, z_1), \cdots, T_m'(x, z_m))$, which occur whenever $z_i = 1$, by adding $|z|$ one votes. Since we took $|y| = (|z|+m)/2$, then $|y| - (m-|y|) = |z|$, so adding $y_1, \cdots, y_{m}$ to the majority vote gives a surplus of $|z|$ one votes. 

    Therefore, we can see that the distribution over functions $\eval(z_1 \cdots z_m, y_1 \cdots y_{m}, \cdot)$ for $(z,y) \sim \mu_L^*$ is, when considering functional equivalence, identical to $\mu_L$.
    Since $\mu_L^*$ is obviously polynomial time samplable, and $\eval \in \majthr$, the theorem statement follows by \thmref{mainResult}.
\end{proof}

\paragraph{Non-black-box access to $\mu_L$.} The distPAC-learning algorithm does not need descriptions of $\mu$ or $L$ as input. Instead, it suffices to have black-box access to a sampling machine that outputs appropriately distributed polynomial size circuits that are functionally equivalent to the appropriate $\majthr$-circuit, under some fixed encoding scheme. To see this, observe that the randomized predictor $L$ (Figure \ref{fig:predictor}) uses the evaluation rule $\eval$ and a concept representation $\pi_g$ to basically obtain some labels of $g$. But it is not important what the representation is, or how $\phi$ is implemented. For the accuracy of the learning algorithm, it only matters that, it is \textit{possible} to implement $\eval$ by a $\majthr$ circuit (under some concept representation). Hence, the learner can work with some far messier representation of concepts, such as by arbitrary polynomial size circuits.

The fact that the learner does not need to know $\mu$ or $L$ is a good property that enhances the convenience of the learning algorithm. Arguably, knowledge of $\mu$ and $L$ would be a prohibitive assumption in practice; rather, black-box sample access to a more complex and messy form of $\mu_L$ (provided that the learner can still interpret the encoding well enough to evaluate it) is a significantly weaker assumption, corresponding to a scenario where the learner has some non-explicit knowledge about the possible concepts it may encounter.

\subsubsection{Polytopes and DNFs}

In this section we will explain how we can slightly modify the proof for \thmref{distPAC_MajThr2} to also obtain distPAC-learning for natural distributions over polytopes and DNFs. Recall that, in distributional PAC-learning, subclasses are not necessarily distPAC-learnable if their superclass is, since it is possible that the subclass consists of functions that are hard for the superclass distPAC-learning algorithm.

A polytope is an intersection of linear threshold functions. In other words, it is any function that belongs to the class $\comp{\AND}{\thr}$. Therefore, it is easy to see that any polytope is computable by a $\majthr$-circuit. We consider a slight modification of the distribution over $\majthr$-circuits:

\paragraph{Polytope target distribution.} Let $L = (T_1, \cdots T_{m})$ be a list of $m := \poly(n)$ linear threshold functions. Also, let $\mu$ be any $\poly(n)$ time samplable distribution over $\{0,1\}^{m}$. We define the distribution $\mu^{\land}_L$ over polytopes as follows. 
\begin{itemize}
    \item Sample $\theta \sim \mu$.
    \item Output the polytope that is the $\AND$ over each $T_i \in L$ such that $\theta_i = 0$.
\end{itemize}

\begin{theorem}\label{distPAC_Polytope2}
    Let $L = (T_1, \cdots T_{m})$ be any list of $m := \poly(n)$ linear threshold functions over $n$-bit inputs, and let $\mu$ be any $\poly(n)$ time samplable distribution over $\{0,1\}^{m}$. The pair $(\comp{\AND}{\thr}, \mu^{\land}_L)$ is efficiently distPAC-learnable.
\end{theorem}

\begin{proof}
    To complete the proof, we only modify how the evaluation rule $\eval: \{0,1\}^{3m} \times \{0,1\}^n \rightarrow \{-1,1\}$ is defined in the proof of \thmref{distPAC_MajThr2}:
    \[
    \eval(v_1, \cdots v_m, z_1 \cdots z_m, y_1 \cdots y_{m}, x) := \maj (T_1'(x, z_1), \cdots, T_m'(x, z_m), y_1, \cdots, y_{m}, v_1, \cdots v_m)
    \] 
    We define a distribution $\mu^{**}_L$ over binary representations, which samples $(z,y) \sim \mu^*_L$, and then $v$ is sampled uniformly at random from the set of $m$-bit strings with $|v|=|z|/2$. This means that $(m-|v|)-|v| = m - |z|$, that is, there is a surplus of $|m|-z$ zero votes within $v$. We do this in order to convert the top majority-gate into an and-gate, which is accomplished by adding the surplus of $m-|z|$ zero votes, because this forces every $T_i$ for $z_i = 0$ (of which there are $m - |z|$ of them) to evaluate to 1, in order for $\eval$ to evaluate to 1, as desired.

    Therefore, we can see that the distribution over functions $\eval(v,z,y, \cdot)$ for $(v,z,y) \sim \mu_L^{**}$ is, when considering functional equivalence, identical to $\mu^\land_L$.
    Since $\mu_L^{**}$ is obviously polynomial time samplable, and $\eval \in \majthr$, the theorem statement follows by \thmref{mainResult}.
\end{proof}

A DNF can be viewed as a polytope where each of the linear threshold functions simulate ${\sf OR}$ gates. One can simulate ${\sf OR}$ gates by restricting the weights to be either 0 or 1, and setting the threshold $\theta := 1$. 
We therefore consider a modification of the distribution over polytopes:

\paragraph{DNF target distribution.} Let $L = (T_1, \cdots T_{m})$ be a list of $m := \poly(n)$ linear threshold functions, where each weight is either 0 or 1 and the threshold is fixed to $\theta:= 1$. In other words, each $T_i$ implements an ${\sf OR}$
of some selection of variables. Also, let $\mu$ be any $\poly(n)$ time samplable distribution over $\{0,1\}^{m}$. We define the distribution $\mu^{\land\lor}_L$ over DNFs as follows. 
\begin{itemize}
    \item Sample $\theta \sim \mu$.
    \item Output the DNF that is the $\AND$ over each $T_i \in L$ such that $\theta_i = 0$.
\end{itemize}

\begin{theorem}\label{distPAC_DNF2}
    Let $L = (T_1, \cdots T_{m})$ be a list of $m := \poly(n)$ linear threshold functions over $n$-bit inputs, where each weight is either $0$ or $1$ and the threshold is fixed to $\theta:= 1$. Also, let $\mu$ be any $\poly(n)$ time samplable distribution over $\{0,1\}^{m}$. The pair $(\mathrm{DNF}, \mu^{\land\lor}_L)$ is efficiently distPAC-learnable.
\end{theorem}

\begin{proof}
    Using the same argument as in the proof of \thmref{distPAC_Polytope2}, we can define
    \[
    \eval(v_1, \cdots v_m, z_1 \cdots z_m, y_1 \cdots y_{m}, x) := \maj (T_1'(x, z_1), \cdots, T_m'(x, z_m), y_1, \cdots, y_{m}, v_1, \cdots v_m)
    \]  
    Then, because the threshold functions on the bottom layer are defined to simulate ${\sf OR}$ gates, the theorem now follows by the same arguments following the definition of the evaluation rule in \thmref{distPAC_Polytope2}.
\end{proof}

\section{Impossibility of Encoded-Input Weak PRFs}\label{section:NoWeakPRFs}

In this section, we will apply \thmref{mainResult} to prove that weak PRFs augmented with a keyless encoding procedure cannot exist, when evaluating the encoded input on any key is forced to be done by a $\majthr$-circuit.

\subsection{Encoded-Input Weak PRFs}

We define encoded-input weak PRFs. Our definition is the natural relaxation of the definition of encoded-input \textit{strong} PRFs of Boneh et al. \cite{boneh2018exploring}. 

\begin{definition}[Encoded-input weak PRFs]
    Let $\lambda$ be a security parameter, and $n:=n(\lambda),\kappa := \kappa(\lambda), m:= m(n)$ for polynomially bounded functions $n, \kappa, m$. Consider a trio of algorithms ${\sf f}: \{0,1\}^{\kappa} \times \{0,1\}^n \rightarrow \{0,1\}, {\sf gen}: \{1\}^\lambda \rightarrow \{0,1\}^{\kappa}, {\sf enc}: \{0,1\}^{m} \rightarrow \{0,1\}^n$. 
    \begin{itemize}
        \item ${\sf gen}$ is a polynomial time sampling algorithm that given input parameter $\lambda$ in unary and access to random coins $z \in \{0,1\}^{\poly(\kappa)}$ outputs a key $k \in \{0,1\}^{\kappa}$.
        \item ${\sf enc}$ is a polynomial time algorithm that given a representation of an input $r \in \{0,1\}^m$, outputs an encoding $x \in \{0,1\}^n$.
        \item ${\sf f}$ is a polynomial time algorithm that given a key $k$ and the \underline{encoded} input $x \in \{0,1\}^n$, outputs a value ${\sf f}(k,x) = v \in \{0,1\}$.
    \end{itemize} 
    For $t=t(\lambda), \epsilon=\epsilon(\lambda)$, we say that $({\sf f}, {\sf gen}, {\sf enc})$ is a $(t,\epsilon)$-e.i.weak PRF if, for every size $t$ oracle circuit $C$,
    \[
    \left|\Pr_{k \sim {\sf gen}(1^\lambda)}\left[C^{\ExF{{\sf f}}}=1\right]- \Pr_{\psi}\left[C^{\ExF{\psi}}=1\right]\right| \le \epsilon(\lambda)
    \]
    where $\psi: \{0,1\}^n \rightarrow \{0,1\}$ is a uniformly random function, and $\ExF{h}$ is an oracle that returns random points $\langle {\sf enc}(r), h({\sf enc}(r)) \rangle$ for some function $h: \{0,1\}^n \rightarrow \{0,1\}$ and $r \sim U_m$.  
\end{definition}

To measure the complexity of an e.i.weak PRF, we say that $({\sf f}, {\sf gen}, {\sf enc})$ is \textit{evaluated} by a uniform circuit class $\Lambda$, if ${\sf f}\in \Lambda$ (i.e., we ignore complexity of the encoding procedure). This differs from the notion of fixed-key complexity considered elsewhere in the literature for weak PRFs (where we want that, for every $k$, ${\sf f}(k, \cdot) \in \Lambda$). However, we remark that in some cases, arguably it is more important to know what is the circuit complexity of a single circuit that is capable of evaluating many inputs on many unknown keys.

\subsection{Impossibility of Encoded-Input Weak PRFs Evaluated By Majority-of-Threshold}

In the remainder of this section, we will show that there cannot be any e.i.weak PRFs evaluated by $\majthr$-circuits of polynomial size. The proof is a straightforward application of \thmref{mainResult}.

\begin{theorem}[\thmref{intro:NoEIWPRF} restated]
    There exists no encoded-input weak PRF that is evaluated by a $\majthr$-circuit.
\end{theorem}

\begin{proof}
    Let $\eval$ be an evaluation rule, and suppose that for every $n \in \nat$,
    and product distribution $(\mu, \rho)$, $\fG[\eval, n, (\mu, \rho)]$ is $(c, \gamma)$-evaluated. By \thmref{mainResult}, 
    $(\fC, \mu)$ is distPAC-learnable. For a $\poly(n)$-samplable $\mu$ and $\poly(n)$-represented $\fC$, the learning algorithm runs in time 
    polynomial in $n,\gamma^{-1}, \epsilon^{-1}, \delta^{-1}, \eta^{-1}$, and $2^{c}$.

    Now, consider any trio of algorithms, $({\sf f}, {\sf gen}, {\sf enc})$, with $f \in \Lambda$. In order to apply \thmref{mainResult}, we view ${\sf f} \in \Lambda$ as an evaluation rule, ${\sf gen}$ as a $\poly(n)$-samplable target distribution $\mu_{{\sf gen}}$, and ${\sf enc}$ as a sampling algorithm for an example distribution $\rho_{{\sf enc}}$. It follows then that if for the product distribution $(\mu_{{\sf gen}}, \rho_{{\sf enc}})$, $\fG[{\sf f}, n, (\mu_{{\sf gen}}, \rho_{{\sf enc}})]$ is $(c, \gamma)$-evaluated, then $({\sf f}, {\sf gen}, {\sf enc})$ cannot be a e.i.weak PRF with better than $\poly(2^c)$ security.

    Finally, it follows that for any ${\sf f} \in \majthr$, the trio $({\sf f}, {\sf gen}, {\sf enc})$ cannot be an e.i.weak PRF, because we know that $\fG[{\sf f}, n, (\mu_{{\sf gen}}, \rho_{{\sf enc}})]$ is $(O(\log n), 1/\poly(n))$-evaluated (see \thmref{nisanThm}).
\end{proof}

\subsection{Improving Main Theorem Requires Learning Theory Breakthroughs}\label{sec:RequiresBreakthroughs}

Through the lens of weak PRF distinguishers, we now explain how a slight improvement of the runtime dependency of $2^{O(c)}$ stated in \thmref{mainResult} to $2^{O(c^{1/2})}$ is a difficult task.

A famous example of a ``low-complexity'' weak PRF is the ${\sf XOR{\text -}MAJ}$ weak PRF candidate of \cite{blum1993cryptographic}. The authors of \cite{blum1993cryptographic} claimed 30 years ago that ``any method that could even weakly predict [in polynomial time] such functions over a uniform
distribution would require profoundly new ideas.'' At the moment, there is arguably still no reason to believe that such methods will soon be developed.

We define the weak PRF candidate of \cite{blum1993cryptographic}; our argument follows immediately.
The ${\sf XOR{\text -}MAJ}$ weak PRF candidate is the pair of algorithms $({\sf xm}, {\sf xmgen})$ constructed as follows.

\begin{itemize}
    \item ${\sf xmgen}(1^n)$ outputs the key $k = (A,B) \subseteq [n]$, consisting of uniformly random disjoint sets $A,B \in [n]$ of size $\log n$ each. $k = (A,B)$ can be considered a bitstring of length $2\log^2(n)$.

    \item ${\sf xm}(k,x)$ takes as input the key and a string $x \in \{0,1\}^n$. ${\sf xm}(k,x)$ is defined by
    \[
    {\sf xm}(k,x) = {\sf XOR}({\sf XOR}(x|_A), \maj(x|_B))
    \]
    Here, $x|_S$ is the projection of a string $x \in \{0,1\}^n$ to the coordinates indicated by a set $S \subseteq [n]$.
    
    \end{itemize}

From the definition, it is clear that ${\sf xm}$ can be seen as an evaluation function that induces a $2\log^2(n)$-represented concept class. The target distribution would be uniform over the representations contained in $\{0,1\}^{2\log^2(n)}$. Therefore, \thmref{mainResult}, if improved to run in time polynomial in $2^{c^{1/2}}$, would imply a polynomial time distinguisher for the ${\sf XOR{\text -}MAJ}$ weak PRF candidate. This is justified by the fact that the communication complexity of the associated communication game $\fG[{\sf xm}, n, (U_{2\log^2(n)}, U_n)]$ is at most $2\log^2(n)$ (the first player merely sends its entire input to the second player, who then evaluates the ${\sf xm}$ ans outputs the result), and $2^{O((2\log^2(n))^{1/2})} = \poly(n)$.

\section*{Acknowledgements}
I thank Mark Bun, Ran Canetti, Russell Impagliazzo, and Emanuele Viola for thoughtful conversations about this research. I also thank Mauricio Karchmer for advice on presentational aspects of this paper. Finally, I give special thanks to Marco Carmosino for helpful comments on a draft of this paper, as well as many discussions pertaining to this research. 
Part of this research was completed while I was visiting the Simons Institute for the theory of computing. 

\bibliographystyle{alpha}
\bibliography{bib.bib}

\newcommand{\etalchar}[1]{$^{#1}$}
\begin{thebibliography}{KLMY19}

\bibitem[AB09]{arora2009computational}
Sanjeev Arora and Boaz Barak.
\newblock {\em Computational complexity: a modern approach}.
\newblock Cambridge University Press, 2009.

\bibitem[BCG{\etalchar{+}}21]{boyle2021low}
Elette Boyle, Geoffroy Couteau, Niv Gilboa, Yuval Ishai, Lisa Kohl, and Peter
  Scholl.
\newblock Low-complexity weak pseudorandom functions in ${\sf ac}^0[{\sf
  mod2}]$.
\newblock In {\em Annual International Cryptology Conference}, pages 487--516.
  Springer, 2021.

\bibitem[BFKL93]{blum1993cryptographic}
Avrim Blum, Merrick Furst, Michael Kearns, and Richard~J Lipton.
\newblock Cryptographic primitives based on hard learning problems.
\newblock In {\em Annual International Cryptology Conference}, pages 278--291.
  Springer, 1993.

\bibitem[BIP{\etalchar{+}}18]{boneh2018exploring}
Dan Boneh, Yuval Ishai, Alain Passel{\`e}gue, Amit Sahai, and David~J Wu.
\newblock Exploring crypto dark matter: New simple prf candidates and their
  applications.
\newblock In {\em Theory of Cryptography: 16th International Conference, TCC
  2018, Panaji, India, November 11--14, 2018, Proceedings, Part II}, pages
  699--729. Springer, 2018.

\bibitem[BNS92]{babai1992multiparty}
L{\'a}szl{\'o} Babai, Noam Nisan, and M{\'a}ri{\'o} Szegedy.
\newblock Multiparty protocols, pseudorandom generators for logspace, and
  time-space trade-offs.
\newblock {\em Journal of Computer and System Sciences}, 45(2):204--232, 1992.

\bibitem[Che18]{chen2018toward}
Lijie Chen.
\newblock Toward super-polynomial size lower bounds for depth-two threshold
  circuits.
\newblock {\em arXiv preprint arXiv:1805.10698}, 2018.

\bibitem[CIKK16]{carmosino2016learning}
Marco~L Carmosino, Russell Impagliazzo, Valentine Kabanets, and Antonina
  Kolokolova.
\newblock Learning algorithms from natural proofs.
\newblock In {\em 31st Conference on Computational Complexity (CCC 2016)}.
  Schloss Dagstuhl-Leibniz-Zentrum fuer Informatik, 2016.

\bibitem[CT93]{chung1993communication}
Fan~RK Chung and Prasad Tetali.
\newblock Communication complexity and quasi randomness.
\newblock {\em SIAM Journal on Discrete Mathematics}, 6(1):110--123, 1993.

\bibitem[DSS16]{daniely2016complexity}
Amit Daniely and Shai Shalev-Shwartz.
\newblock Complexity theoretic limitations on learning dnf’s.
\newblock In {\em Conference on Learning Theory}, pages 815--830. PMLR, 2016.

\bibitem[DW{\etalchar{+}}00]{domingo2000madaboost}
Carlos Domingo, Osamu Watanabe, et~al.
\newblock Madaboost: A modification of adaboost.
\newblock In {\em COLT}, pages 180--189, 2000.

\bibitem[Fei02]{feige2002relations}
Uriel Feige.
\newblock Relations between average case complexity and approximation
  complexity.
\newblock In {\em Proceedings of the thiry-fourth annual ACM symposium on
  Theory of computing}, pages 534--543, 2002.

\bibitem[FX14]{feldman2014sample}
Vitaly Feldman and David Xiao.
\newblock Sample complexity bounds on differentially private learning via
  communication complexity.
\newblock In {\em Conference on Learning Theory}, pages 1000--1019. PMLR, 2014.

\bibitem[GGM86]{goldreich1986construct}
Oded Goldreich, Shafi Goldwasser, and Silvio Micali.
\newblock How to construct random functions.
\newblock {\em Journal of the ACM (JACM)}, 33(4):792--807, 1986.

\bibitem[GK23]{goldberg2023improved}
Halley Goldberg and Valentine Kabanets.
\newblock Improved learning from kolmogorov complexity.
\newblock {\em ECCC Report}, 2023.

\bibitem[HILL99]{haastad1999pseudorandom}
Johan H{\aa}stad, Russell Impagliazzo, Leonid~A Levin, and Michael Luby.
\newblock A pseudorandom generator from any one-way function.
\newblock {\em SIAM Journal on Computing}, 28(4):1364--1396, 1999.

\bibitem[JLSW11]{jackson2011learning}
Jeffrey~C Jackson, Homin~K Lee, Rocco~A Servedio, and Andrew Wan.
\newblock Learning random monotone dnf.
\newblock {\em Discrete Applied Mathematics}, 159(5):259--271, 2011.

\bibitem[JS05]{jackson2005learning}
Jeffrey~C Jackson and Rocco~A Servedio.
\newblock Learning random log-depth decision trees under uniform distribution.
\newblock {\em SIAM Journal on Computing}, 34(5):1107--1128, 2005.

\bibitem[Kar23]{karchmer2023agnostic}
Ari Karchmer.
\newblock Agnostic membership query learning with nontrivial savings: New
  results, techniques.
\newblock {\em arXiv preprint arXiv:2311.06690}, 2023.

\bibitem[KLMY19]{kane2019communication}
Daniel Kane, Roi Livni, Shay Moran, and Amir Yehudayoff.
\newblock On communication complexity of classification problems.
\newblock In {\em Conference on Learning Theory}, pages 1903--1943. PMLR, 2019.

\bibitem[KN96]{kushilevitz1996communication}
Eyal Kushilevitz and Noam Nisan.
\newblock Communication complexity, 1996.

\bibitem[KNR99]{kremer1999randomized}
Ilan Kremer, Noam Nisan, and Dana Ron.
\newblock On randomized one-round communication complexity.
\newblock {\em Computational Complexity}, 8:21--49, 1999.

\bibitem[KOS04]{klivans2004learning}
Adam~R Klivans, Ryan O'Donnell, and Rocco~A Servedio.
\newblock Learning intersections and thresholds of halfspaces.
\newblock {\em Journal of Computer and System Sciences}, 68(4):808--840, 2004.

\bibitem[KS09]{klivans2009cryptographic}
Adam~R Klivans and Alexander~A Sherstov.
\newblock Cryptographic hardness for learning intersections of halfspaces.
\newblock {\em Journal of Computer and System Sciences}, 75(1):2--12, 2009.

\bibitem[KV94]{kearns1994cryptographic}
Michael Kearns and Leslie Valiant.
\newblock Cryptographic limitations on learning boolean formulae and finite
  automata.
\newblock {\em Journal of the ACM (JACM)}, 41(1):67--95, 1994.

\bibitem[LS09]{linial2009learning}
Nati Linial and Adi Shraibman.
\newblock Learning complexity vs communication complexity.
\newblock {\em Combinatorics, Probability and Computing}, 18(1-2):227--245,
  2009.

\bibitem[Nan21]{nanashima2021theory}
Mikito Nanashima.
\newblock A theory of heuristic learnability.
\newblock In {\em Conference on Learning Theory}, pages 3483--3525. PMLR, 2021.

\bibitem[Nis93]{nisan1993communication}
Noam Nisan.
\newblock The communication complexity of threshold gates.
\newblock {\em Combinatorics, Paul Erdos is Eighty}, 1:301--315, 1993.

\bibitem[Raz87]{razborov1987lower}
Alexander~A Razborov.
\newblock Lower bounds on the size of bounded depth circuits over a complete
  basis with logical addition.
\newblock {\em Mathematical Notes of the Academy of Sciences of the USSR},
  41(4):333--338, 1987.

\bibitem[Raz00]{raz2000bns}
Ran Raz.
\newblock The bns-chung criterion for multi-party communication complexity.
\newblock {\em Computational Complexity}, 9(2):113--122, 2000.

\bibitem[Reg04]{regev2004new}
Oded Regev.
\newblock New lattice-based cryptographic constructions.
\newblock {\em Journal of the ACM (JACM)}, 51(6):899--942, 2004.

\bibitem[Reg09a]{regev2009lattices}
Oded Regev.
\newblock On lattices, learning with errors, random linear codes, and
  cryptography.
\newblock {\em Journal of the ACM (JACM)}, 56(6):1--40, 2009.

\bibitem[Reg09b]{regev2009complexity}
Oded Regev.
\newblock On the complexity of lattice problems with polynomial approximation
  factors.
\newblock In {\em The LLL Algorithm: Survey and Applications}, pages 475--496.
  Springer, 2009.

\bibitem[RR97]{razborov1997natural}
Alexander~A Razborov and Steven Rudich.
\newblock Natural proofs.
\newblock {\em Journal of Computer and System Sciences}, 55(1):24--35, 1997.

\bibitem[Sch90]{DBLP:journals/ml/Schapire90}
Robert~E. Schapire.
\newblock The strength of weak learnability.
\newblock {\em Mach. Learn.}, 5:197--227, 1990.

\bibitem[Sel09]{sellie2009exact}
Linda Sellie.
\newblock Exact learning of random dnf over the uniform distribution.
\newblock In {\em Proceedings of the forty-first annual ACM symposium on Theory
  of computing}, pages 45--54, 2009.

\bibitem[Smo87]{smolensky1987algebraic}
Roman Smolensky.
\newblock Algebraic methods in the theory of lower bounds for boolean circuit
  complexity.
\newblock In {\em Proceedings of the nineteenth annual ACM symposium on Theory
  of computing}, pages 77--82, 1987.

\bibitem[Val84]{valiant1984theory}
Leslie~G Valiant.
\newblock A theory of the learnable.
\newblock {\em Communications of the ACM}, 27(11):1134--1142, 1984.

\bibitem[Vio15]{viola2015communication}
Emanuele Viola.
\newblock The communication complexity of addition.
\newblock {\em Combinatorica}, 35:703--747, 2015.

\bibitem[VW07]{viola2007norms}
Emanuele Viola and Avi Wigderson.
\newblock Norms, xor lemmas, and lower bounds for gf (2) polynomials and
  multiparty protocols.
\newblock In {\em Twenty-Second Annual IEEE Conference on Computational
  Complexity (CCC'07)}, pages 141--154. IEEE, 2007.

\bibitem[Yao82]{yao1982theory}
Andrew~C Yao.
\newblock Theory and application of trapdoor functions.
\newblock In {\em 23rd Annual Symposium on Foundations of Computer Science
  (SFCS 1982)}, pages 80--91. IEEE, 1982.

\end{thebibliography}

\end{document}